\documentclass[11pt]{article}
\usepackage{dsfont}
\usepackage{amsmath}
\usepackage{amsthm}
\usepackage{amssymb}
\usepackage{algorithmicx}
\usepackage{algorithm}
\usepackage{algpseudocode}
\usepackage{graphicx}
\usepackage{url}
\usepackage{hyperref}
\usepackage[letterpaper,margin=1in]{geometry}
\usepackage{authblk}
\usepackage{tikz} 
\usepackage{subcaption}
\usepackage{slashbox}
\usetikzlibrary{arrows,decorations.markings}

\usetikzlibrary{shapes.misc}
\usetikzlibrary{calc}

\tikzset{cross/.style={cross out, draw=black, minimum size=2*(#1-\pgflinewidth), inner sep=0pt, outer sep=0pt},
cross/.default={1pt}}

\newcommand{\MyFrame}[1]{\begin{center}\noindent \framebox[\textwidth]{ \begin{minipage}{0.97\textwidth} #1 \end{minipage}}\end{center}}%

\usepackage{color}

\definecolor{blueblack}{rgb}{0,0,.7}

\newcounter{sideremark}

\renewcommand{\setminus}{-}

\theoremstyle{plain}
\newtheorem{theorem}{Theorem}[section]
\newtheorem{lemma}[theorem]{Lemma}

\newtheorem{defn}[theorem]{Definition}

\newtheorem{obs}[theorem]{Observation}
\newtheorem{proposition}[theorem]{Proposition}

\newtheorem{cor}[theorem]{Corollary}
\newtheorem{fact}[theorem]{Fact}

\theoremstyle{remark}
\newtheorem*{rem}{Remark}
\newtheorem*{note}{Note}

\def\eg{{\it e.g.,}~}
\def\ie{{\it i.e.,}~}

\newcommand{\Q}{\mathds{Q}}

\makeatletter
\newtheorem*{rep@theorem}{\rep@title}
\newcommand{\newreptheorem}[2]{%
\newenvironment{rep#1}[1]{%
 \def\rep@title{\emph{\textbf{#2} \ref{##1}}}%
 \begin{rep@theorem}}%
 {\end{rep@theorem}}}
\makeatother
\newreptheorem{theorem}{Theorem}
\newreptheorem{lemma}{Lemma}
\newreptheorem{proposition}{Proposition}
\usepackage{thmtools,thm-restate}

\newcommand{\eps}{\varepsilon}
\newcommand{\opt}{\text{OPT}}

\newcommand{\dist}{\text{dist}}

\newcommand{\cost}{\text{cost}}

\newcommand{\calM}{\mathcal{M}}

\newcommand{\local}{\mathcal{L}}
\newcommand{\globalS}{C^*}
\newcommand{\localt}{\tilde{\mathcal{L}}}
\newcommand{\globalSt}{\tilde{\globalS}}
\newcommand{\IR}{\text{IR}}
\newcommand{\reas}{\text{Reassign}}
\newcommand{\Nloc}{N_{\local}}
\newcommand{\Nglob}{N_{\globalS}}

\newcommand{\Rloc}{R^{\localt}}
\newcommand{\Rglob}{R^{\globalSt}}
\newcommand{\Ngam}{N_{\Gamma}(\ell)}
\newcommand{\clients}{A}
\newcommand{\candidates}{F}

\newif\ifshort
\shorttrue 

\title{On the Local Structure of Stable Clustering Instances}
\author{Vincent Cohen-Addad\thanks{vcohenad@gmail.com}}
\affil{University of Copenhagen}
\author{Chris Schwiegelshohn~\thanks{chris.schwiegelshohn@tu-dortmund.de}}
\affil{Sapienza University of Rome}
\date{}
\begin{document}
\maketitle
\begin{abstract}
We study the classic $k$-median and $k$-means clustering objectives in  the \emph{beyond-worst-case} scenario.
We consider three well-studied notions of structured data that aim at characterizing real-world inputs:
  \begin{itemize}
  \item Distribution Stability (introduced by Awasthi, Blum, and 
    Sheffet, FOCS 2010)
  \item Spectral Separability (introduced by Kumar and Kannan, FOCS 2010)
  \item Perturbation Resilience (introduced by Bilu and Linial, ICS 2010)
  \end{itemize}
We prove structural results showing that inputs satisfying at least one of the conditions are inherently ``local''.
Namely, for any such input, any local optimum is close both in term of structure and in term of objective value to the global optima.

As a corollary we obtain that the widely-used Local Search algorithm has strong performance guarantees for both the tasks of recovering the underlying optimal clustering and obtaining a clustering of small cost.
This is a significant step toward understanding the success of local search heuristics in clustering applications.
\end{abstract}

\newpage
\pagenumbering{arabic}

\setcounter{page}{1}
\section{Introduction}
Clustering is a fundamental, routinely-used approach to extract 
information from datasets. 
Given a dataset and the most important features of the data, a 
clustering
is a partition of the data such that data elements in the same part have 
common features.
The problem of computing a clustering has received
a considerable amount of attention in both practice and theory.

The variety of contexts in which clustering problems arise makes the 
problem of computing a ``good'' clustering hard to define formally.
From a theoretician's perspective, clustering problems are often modeled by 
an objective function we wish to optimize (\eg the 
famous $k$-median or $k$-means objective functions). 
This \emph{modeling step} is both needed and crucial since
it provides a framework to quantitatively compare algorithms.
Unfortunately, the most popular objectives for clustering, like the $k$-median and $k$-means objectives, are 
hard to approximate, even when restricted to Euclidean spaces. 

\medskip

This view is generally not shared by practitioners. Indeed, clustering is often used as a preprocessing step to simplify and speed up subsequent analysis, even if this analysis admits polynomial time algorithms.
If the clustering itself is of independent interest, there are many heuristics with good running times and results on real-world inputs.
\medskip

This induces a gap between theory and practice.
On the one hand, the algorithms that are efficient in practice cannot be proven to achieve good approximation to the $k$-median and $k$-means objectives in the worst-case.
Since approximation ratios are one of the main methods to evaluate algorithms, theory predicts that determining a good clustering is a difficult task.
On the other hand, the best theoretical algorithms turn out to be noncompetitive in applications because they are designed to handle ``unrealistically'' hard instances with little importance for practitioners.
To bridge the gap between theory and practice, it 
is necessary to go \emph{beyond the worst-case analysis} by, for example,
characterizing and focusing on inputs that arise in practice.

\subsection{Real-world Inputs}
Several approaches have been proposed to bridge the gap between theory 
and practice. For example, researchers have considered
the average-case scenario (\eg \cite{ben1992theory}) 
where the running time of an algorithm
is analyzed with respect to some probability distribution 
over the set of 
all inputs. Smooth analysis (\eg \cite{SpT04}) is 
another celebrated approach that
analyzes the running time of an algorithm with respect to worst-case
inputs subject to small random perturbations.

Another successful approach, the one we take in this paper, 
consists in focusing
on \emph{structured} inputs. 
In a seminal paper, Ostrovsky, Rabani, Schulman, and Swamy~\cite{ORSS12} 
introduced the idea that inputs that come from practice induce a 
\emph{ground-truth} or a \emph{meaningful} clustering.
They argued that an input $I$ contains a 
meaningful clustering into $k$ 
clusters if the optimal $k$-median cost of a 
clustering using $k$ centers, say $\opt_k(I)$,
is much smaller than the optimal cost of a clustering using 
$k-1$ centers $\opt_{k-1}(I)$. This 
is also motivated by the \emph{elbow method}\footnote{The elbow-method 
  consists in running an (approximation) algorithm for an 
  incrementally increasing number of clusters until the cost 
  drops significantly.
} (see Section~\ref{sec:bibstability} for more details) 
used by practitioners to define the number of clusters.
More formally, an instance $I$ of $k$-median or $k$-means
satisfies the \emph{$\alpha$-ORSS property}
if $\opt_{k}(I) / \opt_{k-1}(I) \le \alpha$.

$\alpha$-ORSS inputs exhibit interesting
properties.
The popular $k$-means$++$ algorithm (also known as the $D^2$-sampling 
technique) achieves an $O(1)$-approximation for these inputs\footnote{
For worst-case inputs, the $k$-means$++$ achieves an $O(\log k)$-approximation 
ratio~\cite{ArV07,BMORST11,JaG12,ORSS12}.}.
The condition is also robust with respect to noisy perturbations of the data set. 
ORSS-stability also implies several other conditions aiming to capture well-clusterable instances.
Thus, the inputs satisfying the ORSS property arguably share some properties with the real-world inputs. 
In this paper, we also provide experimental results 
supporting this claim, see Appendix~\ref{S:experiments}.
\smallskip

These results have opened new research directions and raised several
questions. For example:
\begin{itemize}
\item Is it possible to obtain similar results for more general
  classes of inputs?
\item How does the parameter $\alpha$ impact the
  approximation guarantee and running time?  
\item Is it possible to prove good performance guarantees for other
  popular heuristics? 
\item How close to the ``ground-truth'' clustering are the 
  approximate clusterings?
\end{itemize}
We now review the most relevant work in connection to 
the above open questions, see Sections~\ref{sec:relatedwork}  
for other related work.
\paragraph{Distribution Stability (Def.~\ref{defn:betadelta})}
Awasthi, Blum and Sheffet~\cite{ABS10} have tackled  
the first two questions by introducing the notion of 
\emph{distribution stable} instances. 
Distribution stable instances are a generalization of the ORSS instances
(in other words, any instance satisfying the ORSS property
is distribution stable). They also introduced a new algorithm tailored
for distribution stable instances that achieves a 
$(1+\eps)$-approximation for $\alpha$-ORSS inputs (and more generally
$\alpha$-distribution stable instances) 
in time $n^{O(1/\eps\alpha)}$. This was the first algorithm 
whose approximation guarantee was independent from the parameter
$\alpha$ for $\alpha$-ORSS inputs.

\paragraph{Spectral Separability (Def.~\ref{defn:Ssep})}
Kumar and Kannan~\cite{KuK10} tackled the first and third questions by introducing the \emph{proximity} condition\footnote{
  In this paper, we work with a slightly more general condition called \emph{spectral separability} but the 
  motivations behind the two conditions are similar.}.
This condition also generalizes the ORSS condition.
It is motivated by the goal of learning a distribution mixture in a $d$-dimensional Euclidean space. 
Quoting~\cite{KuK10}, the message of their paper can loosely be stated as: 
\begin{quote}
  If the projection of any data point onto the line joining its cluster center to any other cluster center is $\gamma k$
  times standard deviations closer to its own center than the other center, then we can cluster correctly in polynomial time.
\end{quote}
In addition, they have made a significant step toward understanding the success of the classic \texttt{$k$-means} by
showing that it achieves a $1+O(1/\gamma)$-approximation for instances that satisfy the proximity condition.


\paragraph{Perturbation Resilience (Def.~\ref{defn:perturb})}
In a seminal work, Bilu and Linial~\cite{BiL12} introduced a new 
condition to capture real-world instances.
They argue that the optimal solution of a real-world instance
is often much better than any other solution and so, a slight perturbation of the
instance does not lead to a different optimal solution. Perturbation-resilient
instances have been studied in various contexts
(see \eg~\cite{ABS12,BaC16,BHW16,BaL16,BeR14,KSB16}).
For clustering problems, an instance is said to be \emph{$\alpha$-perturbation resilient} if an adversary 
can change the 
distances between pairs of elements by a factor at most $\alpha$ and the optimal 
solution remains the same.
Recently, 
Angelidakis, Makarychev, and Makarychev~\cite{MM16} have given a polynomial-time 
algorithm for solving 
$2$-perturbation-resilient instances\footnote{We note that it is NP-hard to 
recover the optimal clustering of a $<2$-perturbation-resilient instance~\cite{BeR14}.
}. Balcan and Liang~\cite{BaL16} have
tackled the third question by showing that a classic algorithm for 
hierarchical clustering can solve $1+\sqrt{2}$-perturbation-resilient 
instances. This very interesting result leaves open the question as whether classic
algorithms for (``flat'') clustering could also be proven to be efficient for perturbation-resilient 
instances.

\paragraph{Main Open Questions} Previous work has made important steps 
toward bridging the gap between theory and practice for clustering problems.
However, we still do not have a complete understanding of the properties of ``well-structured'' inputs, nor do we know why the algorithms used in practice perform so well. 
Some of the most important open questions are the following:

\begin{itemize}
\item Do the different definitions of well-structured input have 
  common properties?
\item Do heuristics used in practice have strong approximation ratios for well-structured inputs?
\item Do heuristics used in practice recover the ``ground-truth'' 
  clustering on well-structured inputs?
\end{itemize}


\subsection{Our Results: A unified approach via Local Search}
We make a significant step toward answering the above open questions.
We show that the classic 
Local Search heuristic (see 
Algorithm~\ref{alg:LS}), that has found 
widespread application in practice (see Section~\ref{sec:relatedwork}),
achieves \emph{good} approximation guarantees for distribution-stable, 
spectrally-separable, and perturbation-resilient instances 
(see Theorems~\ref{thm:beta-delta},~\ref{thm:perturbation},~\ref{thm:spectralptas}).

More concretely, we show that Local Search is a 
polynomial-time approximation scheme (PTAS) for both 
distribution-stable and 
spectrally-separable\footnote{Assuming a standard preprocessing step 
consisting of a projection onto a subspace of lower dimension.} 
instances. 
In the case of distribution stability, we also answer the above open question by showing that \emph{most} of the structure of the optimal underlying clustering is recovered by the algorithm.
Furthermore, our results hold even when only a $\delta$ fraction (for any constant $\delta>0$) of the points of each optimal cluster satisfies the $\beta$-distribution-stability property.

For $\gamma$-perturbation-resilient instances, we show that if $\gamma>3$ then any solution is the optimal solution if it cannot be improved by adding or removing $2\gamma$ centers. 
We also show that the analysis is essentially tight.

\MyFrame{
These results show that well-structured inputs have the property that the local optima are close both qualitatively (in terms of structure) and quantitatively (in terms of objective value) to the global ``ground-truth'' optimum.
\medskip

These results make a significant step toward explaining the 
success of Local Search approaches for solving clustering problems
in practice.
}

\begin{algorithm}
  \caption{Local Search($\eps$) for $k$-Median and $k$-Means}
  \label{alg:LS}
  \begin{algorithmic}[1]
    \State \textbf{Input:} $A,F,\cost,k$
    \State \textbf{Parameter:} $\eps$
    \State $S \gets$ Arbitrary subset of $F$ of cardinality at most $k$.
    \While{$\exists$ $S'$ s.t. $|S'|\leq k$  \textbf{and} $|S \setminus S'| + |S' \setminus S| \leq 2/\eps$
      \textbf{and} cost($S'$) $\le$ $(1-\eps/n)$ cost($S$)\\}
    \State $S \gets S'$
    \EndWhile
    \State \textbf{Output:} $S$     
  \end{algorithmic}
\end{algorithm}

\subsection{Organization of the Paper}
Section~\ref{sec:relatedwork} provides a more detailed review of previous work on worst-case approximation algorithms and Local Search. Further comments on stability conditions not covered in the introduction can be found in Section~\ref{sec:bibstability} at the end of the paper.
Section~\ref{sec:prelim} introduces preliminaries and notation. 
Section~\ref{sec:distributionstable} is dedicated to distribution-stable instances, 
Section~\ref{sec:perturbation-resilient}
to perturbation-resilient instances, and Section~\ref{sec:euclid} 
to spectrally-separated instances.
All the missing proofs can be found in the appendix.

\section{Related Work}
\label{sec:relatedwork}
\paragraph{Worst-Case Hardness} The problems we study are NP-hard: $k$-median and $k$-means are already NP-hard in the Euclidean plane  (see Meggido and Supowit~\cite{MeS84}, Mahajan et al. ~\cite{MNV12}, and Dasgupta and Freud~\cite{DaF09}). In terms of hardness of approximation, both problems are APX-hard, even in the Euclidean setting when both $k$ and $d$ are part of the input (see Guha and Khuller~\cite{GuK99}, Jain et al.~\cite{JMS02}, Guruswami et al.~\cite{GI03}, and Awasthi et al.~\cite{ACKS15}). 
On the positive side, constant factor approximations are known in metric space for both  $k$-median and $k$-means (see~\cite{ANSW16,BPRST15,LiS13,JaV01,MeP03}). 
For Euclidean spaces we have a PTAS for both problems, either assuming $d$ fixed and $k$ arbitrary~\cite{ARR98,CAKM16,FRS16a,HaK07,HaM04,KoR07}, or assuming $k$ fixed and $d$ arbitrary~\cite{FeL11,KSS10}.

\paragraph{Local Search}
Local Search is an all-purpose heuristic that may be applied to any problem, see Aarts and Lenstra~\cite{AaL97} for a general introduction.
For clustering, there exists a large body of bicriteria approximations for $k$-median and $k$-means~\cite{BaV15,ChG05,CoM15,KPR00}.
Arya et al.~\cite{AGKMMP04} showed that Local Search with a neighborhood size of $1/\varepsilon$ gives a $3+2\varepsilon$ approximation to $k$-median, see also~\cite{GuT08}.
Kanungo et al.~\cite{KMNPSW04} proved an approximation ratio of $9+\varepsilon$ for $k$-means clustering by Local Search, which was until very recently~\cite{ANSW16} the best known algorithm with a polynomial running time in metric and Euclidean spaces.\footnote{They combined Local Search with techniques from Matousek~\cite{Mat00} for $k$-means clustering in Euclidean spaces. The running time of the algorithm as stated incurs an additional factor of $\varepsilon^{-d}$ due to the use of Matousek's approximate centroid set. Using standard techniques (see e.g. Section~\ref{sec:euclidperturb} of this paper), a fully polynomial running time in $n$, $d$, and $k$ is also possible without sacrificing approximation guarantees.} 
Recently, Local Search with an appropriate neighborhood size was shown to be a PTAS for $k$-means and $k$-median in certain restricted metrics including constant dimensional Euclidean space~\cite{CAKM16,FRS16a}.
Due to its simplicity, Local Search is also a popular subroutine for clustering tasks in various more specialized computational models~\cite{BBLM14,BeT10,GMMMO03}.
For more theoretical clustering papers using Local Search, we refer to~\cite{SAD04,DGK02,FrZ16,HaM01,YSZUC08}.

Local Search is also often used for clustering in more  applied areas of computer science (\eg \cite{tuzun1999,Ghosh2003150,Ardjmand201432,hansen2001variable}).
Indeed, the use of Local Search with a neighborhood of size $1$ for clustering was first proposed by T\"uz\"un and Burke~\cite{TuB99}, see also Ghosh~\cite{Gho03} for a more efficient version of the same approach.
Due the ease by which it may be implemented, Local Search has become one of the most commonly used heuristics for clustering and facility location, see Ardjmand~\cite{APWA14}. 
Nevertheless, high running times is one of the biggest drawbacks of Local Search compared to other approaches, though a number of papers have engineered it to become surprisingly competitive, see Frahling and Sohler~\cite{FrS08}, Kanungo et al.~\cite{KMNPSW02}, and Sun~\cite{Sun06}.

\section{Definitions and Notations}
\label{sec:prelim}

\paragraph{The problem} The problem we consider in this work is the following
slightly more general version of the $k$-means and $k$-median problems.
\begin{defn}[$k$-Clustering]
Let $A$ be a set of clients, $F$ a set of centers, both lying in a metric space $(\mathcal{X},\dist)$,
 $\text{\em cost}$ a function $A\times F\rightarrow \mathbb{R}_+$, and $k$ a non-negative integer.
The \emph{$k$-clustering problem} asks for a subset $S$ of $F$, of cardinality at most $k$, that minimizes 
  \[\text{cost}(S)= \sum_{x \in A} \min_{c \in S} \text{cost}(x,c).\]
The clustering of $A$ \emph{induced by $S$} is the partition of $A$ into 
subsets $C=\{C_1,\ldots C_k\}$ such that $C_i=\{x\in A~|~c_i = \underset{c\in S}{\text{argmin }}\text{cost}(x,c)\}$
(breaking ties arbitrarily).
\end{defn}
The well known $k$-median and $k$-means problems correspond to the special 
cases $\text{cost}(a,c)=\dist(a,c)$ and $\text{cost}(a,c)=\dist(a,c)^2$ respectively.
Throughout the rest of this paper, let $\opt$ denote the value of an optimal solution.
To give slightly simpler proofs for $\beta$-distribution-stable and $\alpha$-perturbation-resilient instances,
we will assume that $\cost(a,b) = \dist(a,b)$. 
If $\cost(a,b)=\dist(a,b)^p$, then $\alpha$ depends exponentially on the $p$ for perturbation resilience.
For distribution stability, we still have a PTAS by introducing a dependency in $1/\eps^{O(p)}$ in the neighborhood size of the algorithm. 
The analysis is unchanged save for various applications of the following lemma at different steps of the proof.

\begin{lemma}\label{lem:TI}
  Let $p \ge 0$ and $1/2 > \eps>0$. For any $a,b,c \in A \cup F$, we have
  $\cost(a,b) \le (1+\eps)^p \cost(a,c) + \cost(c,b)(1+1/\eps)^p$.
\end{lemma}

\section{Distribution Stability}
\label{sec:distributionstable}
We work with the notion of $\beta,\delta$-distribution stability which generalizes 
$\beta$-distribution stability.
This extends our result to datasets that exhibit a slightly weaker structure than
the $\beta$-distribution stability.
Namely, the $\beta,\delta$-distribution stability only requires that for
each cluster of the optimal solution, most of the 
points satisfy the $\beta$-distribution stability condition. 

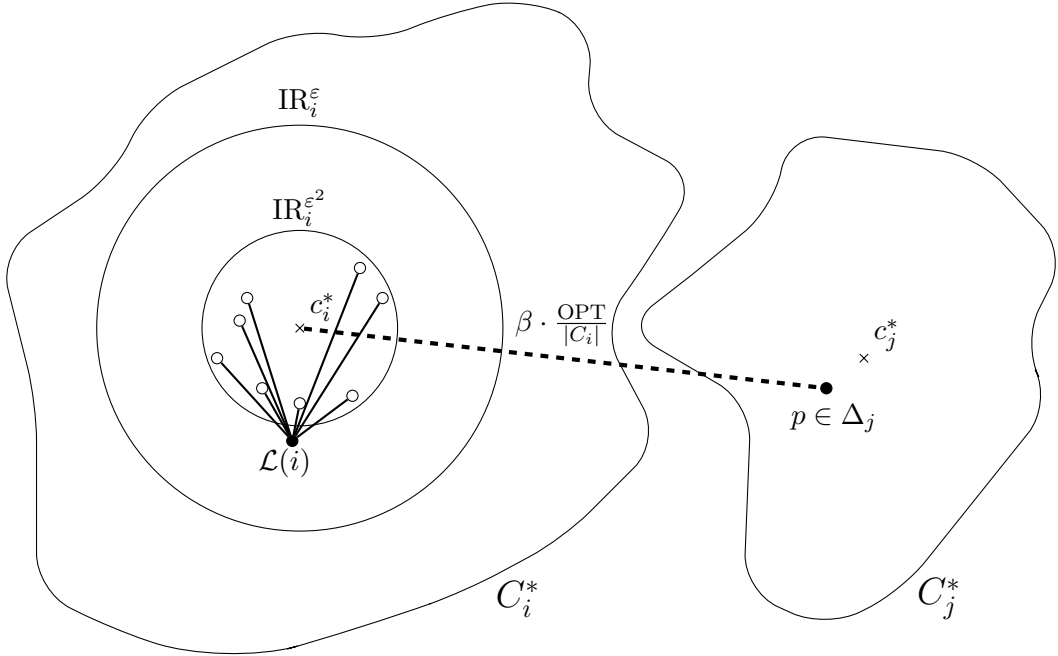
\begin{figure}
\begin{center}
\begin{tikzpicture}[rotate = 0]


\draw [rounded corners=0.5cm] (0.5,5.9) -- (0,6) -- (-2,5) -- (-2.5,4) -- (-4,3) -- (-3.5,1) -- (-3.5,-1.5)  -- (-2.5, -2) [rounded corners=1cm] -- (-1,-2.5) -- (0.8, -2) [rounded corners=0.5cm] -- (2.2, -1.5) -- (3.5, -0.8) -- (4.8,0.5) -- (4, 2) -- (4.7, 3) -- (5.3, 4) -- (3.8, 4.8) -- (3.9, 6) -- (3, 6.4) -- (2,6.2) -- (1, 5.8) -- (0.5,5.9);

\draw [rounded corners = 0.5cm] (5.2,2.7) -- (4.3,2) -- (6,1) -- (5.9, -1.5) [rounded corners = 0.5cm] -- (7,-2) -- (8, -1.5) -- (10,1) -- (9.6, 1.8) -- (10.2,3) -- (9, 4.3)  -- (6.5, 4.6) -- (6.3, 3.6) -- (5.2,2.7) ;

\node[cross = 2pt, label={[xshift=0.3cm, yshift=-0.1cm]$c_i^*$}] (c1) at (0,2) {};
\node[draw=black,fill=black, shape=circle, inner sep = 1.5pt, label={[xshift=-0.1cm, yshift=-0.7cm]$\local(i)$}] (l) at (-0.1,0.5) {};
\draw (0,2) circle (1.3cm);
\draw (0,2) circle (2.7cm);
\node (x) at (0,3.6) {$\IR^{\eps^2}_i$};
\node (x) at (0,5) {$\IR^{\eps}_i$};
\node (x) at (2.9,-1.6) {\Large $C_i^*$};
\node (x) at (8.5,-1.6) {\Large $C_j^*$};

\node[draw=black,fill=none, shape=circle, inner sep = 1.5pt] (a1) at (0,1) {};
\node[draw=black,fill=none, shape=circle, inner sep = 1.5pt] (a2) at (-0.5,1.2) {};

\node[draw=black,fill=none, shape=circle, inner sep = 1.5pt] (a3) at (-1.1,1.6) {};
\node[draw=black,fill=none, shape=circle, inner sep = 1.5pt] (a4) at (-0.7,2.4) {};
\node[draw=black,fill=none, shape=circle, inner sep = 1.5pt] (a5) at (-0.8,2.1) {};

\node[draw=black,fill=none, shape=circle, inner sep = 1.5pt] (a6) at (0.7,1.1) {};

\node[draw=black,fill=none, shape=circle, inner sep = 1.5pt] (a7) at (1.1,2.4) {};
\node[draw=black,fill=none, shape=circle, inner sep = 1.5pt] (a8) at (0.8,2.8) {};

\draw[thick] (l) -- (a1);
\draw[thick] (l) -- (a2);
\draw[thick] (l) -- (a3);
\draw[thick] (l) -- (a4);
\draw[thick] (l) -- (a5);
\draw[thick] (l) -- (a6);
\draw[thick] (l) -- (a7);
\draw[thick] (l) -- (a8);

\node[cross = 2pt, label={[xshift=0.3cm, yshift=-0.1cm]$c_j^*$}] (c2) at (7.5,1.6) {};
\node[draw=black,fill=black, shape=circle, inner sep = 1.5pt, label={[xshift=0.1cm, yshift=-0.8cm]$p\in \Delta_j$}] (p) at (7,1.2) {};

\draw[ultra thick,dashed] (c1) --  node[above] {$\beta\cdot\frac{\opt}{|C_i|}$} (p);

\end{tikzpicture}
\caption{Example of a cluster $C^*_i \not \in Z^*$. An important fraction of the points in $\IR^{\eps^2}_i$ are
  served by $\local(i)$ and few points in $\bigcup_{j \neq i} \Delta_j$ are served by $\local(i)$.}
\label{fig:local}
\end{center}
\end{figure}

\begin{defn}[$(\beta,\delta)$-Distribution Stability]
\label{defn:betadelta}
   Let $(A,F,\text{cost},k)$ be an instance of $k$-clustering where $A\cup F$ lie in a metric space and let 
   $S^*=\{c^*_1,\ldots, c^*_k\} \subseteq F$ be a set of centers and $C^*=\{C^*_1,\ldots ,C^*_k\}$ be
   the clustering induced by $S^*$.
   Further, let $\beta>0$ and $0\le\delta\le 1$. 
   Then the pair $(A,F,\text{cost},k),(C^*,S^*)$ is a 
   \emph{$(\beta,\delta)$-distribution stable instance} if, for any $i$, 
   there exists a set $\Delta_i \subseteq C^*_i$ such that 
   $|\Delta_i| \ge (1-\delta)|C^*_i|$ 
   and for any $x \in \Delta_i $, for any $j \neq i$,
   $$
   \cost(x,c^*_j) \ge \beta \frac{\opt}{|C^*_j|},
   $$
   where $\cost(x,c^*_j)$ is the cost of assigning $x$ to $c^*_j$.
 \end{defn}
For any instance $(A,F,\text{cost},k)$ that is $(\beta,\delta)$-distribution stable, 
we refer to $(C^*,S^*)$ as a $(\beta,\delta)$-clustering of the instance.
We show the following theorem for the $k$-median problem. 
For the $k$-clustering problem with parameter $p$, 
the constant $\eta$ becomes a function of $p$.
\begin{theorem}
  \label{thm:beta-delta}
  Let $p >0$, 
$\beta>0$, and $\eps < \min(1-\delta,1/3)$. 
  For a $(\beta,\delta)$-stable instance with $(\beta,\delta)$ clustering
  $(C^*,S^*)$ and an absolute constant $\eta$, the cost of the 
  solution output by Local Search$(4\eps^{-3 }\beta^{-1}+O(\eps^{-2}\beta^{-1}))$
  (Algorithm \ref{alg:LS})
  is at most $(1+ \eta\eps) \cost(C^*)$.

  Moreover,
  let $L = \{L_1,\ldots,L_k\}$ denote the clusters of the solution output by 
  Local Search$(4\eps^{-3 }\beta^{-1}+O(\eps^{-2}\beta^{-1}))$. 
  If $\delta=0$ (i.e.: the instance is simply $\beta$-distribution-stable), 
  there exists a bijection $\phi: L \mapsto C^*$ 
  such that for at least $m = k - O( \eps^{-3}\beta^{-1})$ clusters 
  $L'_1,\ldots,L'_m \subseteq L$, 
  the following two statements hold.
  \begin{itemize}
    \item At least a $(1-\eps)$ fraction of $\IR^{\eps^2}_i\cap C^*_i$ 
    are served by a unique center $\local(i)$ in solution $\local$.
  \item The total number of clients $p \in \bigcup_{j \neq i} C^*_j$ 
     served by $\local(i)$ in $\local$ is at most
    $\eps  |\IR^{\eps^2}_i\cap C^*_i|$.   
  \end{itemize}
\end{theorem}



We first give a high-level description of the analysis.
Assume for simplicity that all the optimal clusters cost less than an $\eps^3$ fraction of 
the total cost of the optimal solution.
Combining this assumption with the $\beta$-distribution-stability property, one can show that the centers and points close to the center are far away from each other. 
Thus, guided by the objective function, the local search algorithm identifies most of these centers. 
In addition, we can show that for most of these good centers the corresponding cluster in the local solution is very similar to the optimal cluster (see Figure~\ref{fig:local}). 
In total, only very few clusters (a function of $\eps$ and $\beta$)
of the optimal solution are not present in the local solution. 
We conclude our proof by using local optimality.
Our proof includes a few ingredients from~\cite{ABS10} such as the notion of \emph{inner-ring} (we work with a slightly
more general definition) and distinguishes between \emph{cheap} and \emph{expensive} clusters. 
Nevertheless our analysis is slightly stronger as we consider a significantly weaker stability condition and can not only analyze the cost of the solution of the algorithm, but also the structure of its clusters.

Throughout this section, we consider a set of centers
$S^* = \{c^*_1,\ldots,c^*_k\}$ whose induced clustering is $C^* = \{C^*_1,\ldots,C^*_k\}$
and such that the instance is $(\beta,\delta)$-stable with respect $(C^*,S^*)$.
We denote by \emph{clusters} the parts of a partition 
$C^*=\{C^*_1,\ldots ,C^*_k\}$.
Let $\cost(C^*) = \sum_{i=1}^k \sum_{x \in C^*_i} \cost(x,c^*_i)$. 
Moreover, for any cluster $C^*_i$, for any client $x \in C^*_i$, denote by $g_x$ the cost of client $x$ in solution $C^*$: 
$g_x = \cost(x,c^*_i) = \dist(x,c^*_i)$ since we consider the $k$-median problem.
Let $\local$ denote the output of LocalSearch($\beta^{-1} \eps^{-3}$) and $l_x$ the 
cost induced by client $x$ in solution $\local$, namely $l_x = \min_{\ell \in \local}\cost(x,\ell)$, and $\cost(\local)=\sum_{x\in A}l_x$.
The following definition is a generalization of the inner-ring definition of~\cite{ABS10}.

\begin{defn}
  For any $\eps_0$, we define the \emph{inner ring} of cluster $i$, $\IR^{\eps_0}_i$,
  as the set of $x \in A \cup F $ such that 
  $
  \dist(x,c^*_i) \le \eps_0 \beta {\opt}/{|C^*_i|}
  $.
\end{defn}

We say that cluster $i$ is \emph{cheap} if $\sum_{x \in C^*_i} g_x \le \eps^3 \beta \opt$, and \emph{expensive} otherwise.
We aim at proving the following structural lemma.
\begin{lemma}\label{lem:structbeta}
  There exists a set of clusters $Z^* \subseteq C^*$ of size
  at most $2\eps^{-3}\beta^{-1} + O(\eps^{-2}\beta^{-1})$ such that 
  for any cluster $C^*_i \in C^* \setminus Z^*$, we have
  the following properties
  \begin{enumerate}
  \item $C^*_i$ is cheap.
  \item At least a $(1-\eps)$ fraction of $\IR^{\eps^2}_i\cap C^*_i$ 
    are served by a unique center $\local(i)$ in solution $\local$.
  \item The total number of clients $p \in \bigcup_{j \neq i} \Delta_j$ 
     served by $\local(i)$ in $\local$ is at most
    $\eps  |\IR^{\eps^2}_i\cap C^*_i|$.
  \end{enumerate}
\end{lemma}

See Fig~\ref{fig:local} for a typical cluster of $C^* \setminus Z^*$.
We start with the following lemma which generalizes Fact 4.1 in \cite{ABS10}.
\begin{lemma}\label{lem:sizeIR}
  Let $C^*_i$ be a cheap cluster. For any $\eps_0$,
we have $|\IR^{\eps_0}_i\cap C_i^*| > (1-{\eps^3}/{\eps_0}) |C^*_i|$.
\end{lemma}

We then prove that the inner rings of cheap clusters are disjoint for $\delta + \frac{\eps^3}{\eps_0} < 1$ and $\eps_0<\frac{1}{3}$.
\begin{lemma}\label{lem:IRinter}
  Let $\delta + \frac{\eps^3}{\eps_0} < 1$ and $\eps_0<\frac{1}{3}$. If $C^*_i \neq C^*_j$ are cheap clusters, then $\IR^{\eps_0}_i \cap \IR^{\eps_0}_j = \emptyset$.
\end{lemma}



For each cheap cluster $C^*_i$, let  $\local(i)$ denote a center of $\local$ that
belongs to $\IR^{\eps}_i$ if there exists exactly such center and remain undefined otherwise. 
By Lemma~\ref{lem:IRinter}, $\local(i) \neq \local(j)$ for $i \neq j$.

\begin{lemma}\label{lem:nbin-delta}
 Let $\eps<\frac{1}{3}$. Let $C^* \setminus Z_1$ denote the set of clusters  $C^*_i $ that are cheap, such that $\local(i)$ is defined and such that at least 
$(1-\eps) |\IR^{\eps^2}_i\cap C_i^*|$ clients of $\IR^{\eps^2}_i\cap C_i^*$ are
served in $\local$ by $\local(i)$. Then $|Z_1|\leq (2\eps^{-3} + 11.25\cdot\eps^{-2}+22.5\cdot\eps^{-1}) \beta^{-1}$.
  %
\end{lemma}
\begin{proof}
There are five different types of clusters in $C^*$:
\begin{enumerate}
\item $k_1$ expensive clusters
\item $k_2$ cheap clusters with no center of $\local$ belonging to $\IR^{\eps}_i $
\item $k_3$ cheap clusters with at least two centers of $\local$ belonging to $\IR^{\eps}_i $
\item $k_4$ cheap clusters with $\local(i)$ being defined and less than $(1-\eps) |\IR^{\eps^2}_i\cap C_i^*|$ clients of $\IR^{\eps^2}_i\cap C_i^*$ are served in $\local$ by $\local(i)$  
\item $k_5$ cheap clusters with $\local(i)$ being defined and at least $(1-\eps) |\IR^{\eps^2}_i\cap C_i^*|$ clients of $\IR^{\eps^2}_i\cap C_i^*$ are served in $\local$ by $\local(i)$
\end{enumerate}
The definition of cheap clusters immediately yields $k_1\le \eps^{-3} \beta^{-1}$.

 Since $\local$ and $C^*$ both have $k$ clusters and the inner rings of cheap clusters are disjoint (Lemma~\ref{lem:IRinter}), we have $c_1 k_1+ c_3 k_3 + k_4+k_5= k_1+k_2+k_3+k_4+k_5=|Z_1|+k_5=k$ with $c_1\ge 0$ and $c_3\ge 2$ resulting in $k_3\leq (c_3-1)k_3=(1-c_1)k_1+k_2\le k_1+k_2$. 

Before bounding $k_2$ and $k_4$, we discuss the impact of a cheap cluster $C^*_i$ with at least a $p$ fraction of the clients of $\IR^{\eps^2}_i\cap C_i^*$ being served in $\local$ by some centers that are not in $\IR^{\eps}_i$.
  By the triangular inequality, the cost for any client $x$ of this $p$ fraction is 
  at least $(\eps-\eps^2) \beta \cost(C^*)/|C^*_i|$.
  Then the total cost of all clients of this $p$ fraction in $\local$ is at least $p|\IR^{\eps^2}_i\cap C_i^*| (1-\varepsilon)\eps \beta \cost(C^*) /|C^*_i|$.
  By Lemma \ref{lem:sizeIR}, substituting $|\IR^{\eps^2}_i\cap C^*_i|$ yields for this total cost
  $$
  p|\IR^{\eps^2}_i\cap C_i^*|(1-\varepsilon) \eps \beta \frac{\cost(C^*)}{|C^*_i|} \ge p 
  (1-\eps)^2 |C^*_i| \eps \beta \frac{\cost(C^*)}{|C^*_i|} = p(1-\eps)^2 \eps \beta \cost(C^*).
  $$
To determine $k_2$, we must use $p=1$ while we have $p>\eps$ for $k_4$. Therefore, the total costs of all clients of the $k_2$ and the $k_4$ clusters in $\local$ are at least $k_2(1-\eps)^2 \eps \beta \cost(C^*)$ and $k_4(1-\eps)^2 \eps^2 \beta \cost(C^*)$, respectively. 

  Now, since $\cost(\local)\leq 5\opt\leq 5\cost(C^*)$, we have $ (k_2+k_4\eps)\eps\beta\le 5/(1-\eps)^2\le 45/4$. 
  
Therefore, we have $|Z_1|=k_1+k_2+k_3+k_4\le 2k_1+2k_2+k_4\le (2\eps^{-3} + 11.25\cdot\eps^{-2}+22.5\cdot\eps^{-1})\beta^{-1}$.
\end{proof}


We continue with the following lemma, whose proof relies on similar 
arguments.
\begin{lemma}\label{lem:nbout-delta}
  There exists a set $Z_2 \subseteq C^* \setminus Z_1$ of size at most 
  $11.25 \eps^{-1}\beta^{-1}$ such that for any cluster
  $C^*_j \in C^* \setminus Z_2$, the total number of clients $x \in 
  \bigcup_{i \neq j} \Delta_i$, that are served by $\local(j)$ in 
  $\local$ is at most $\eps  |\IR^{\eps^2}_i\cap C_i^*|$.
\end{lemma}
Therefore, the proof of Lemma \ref{lem:structbeta} follows from
combining Lemmas \ref{lem:nbin-delta} and 
\ref{lem:nbout-delta}.

We now turn to the analysis of the cost of $\local$.
Let $C(Z^*)  = \bigcup_{C^*_i \in Z^*} C^*_i$.
For any cluster $C^*_i \in C^* \setminus Z^*$, 
let $\local(i)$ be the unique center of $\local$ that serves at least $(1-\varepsilon)| \IR^{\eps^2}_i\cap C_i^*|>(1-\varepsilon)^2|C_i|$ clients of $\IR^{\eps^2}_i\cap C^*_i$, see Lemmas~\ref{lem:structbeta} and~\ref{lem:sizeIR}.
Let $\widehat \local =  \bigcup_{C^*_i \in C^* \setminus Z^*} \local(i)$ and 
define 
$\widehat A$ to be the set of clients that are 
served in solution $\local$ by centers of 
$\widehat \local$.
Finally, let $A(\local(i))$ be the set of clients that are served 
by $\local(i)$ in solution $\local$. Observe
that the $A(\local(i))$ partition $\widehat A$.

\begin{lemma}\label{lem:costbad-delta}
  We have
  $$- {\eps}\cdot \cost(\local)/n + \sum\limits_{x \in \widehat{A}\setminus C(Z^*)} l_x  
  \le \sum\limits_{x \in \widehat{A} \setminus C(Z^*)} g_x
  + \frac{2\varepsilon }{(1-\varepsilon)^2} \cdot (\cost(C^*)+\cost(\local) ).$$
\end{lemma}
\begin{proof}
  Consider the following mixed solution $\calM = \widehat \local \cup \{c_i^*~|~C_i^*\in Z^*\}$.
  We start by bounding the cost of $\calM$. 
  For any client $x \in \widehat{A}$, 
  the center that serves it in $\local$ belongs to $\calM$. Thus
  its cost in $\calM$ is at most $l_x$.
  Now, for any client $x\in C(Z^*)$, the center that serves it in $C^*$ is in $\calM$, so its cost in $\calM$ is at most $g_x$.

  Finally, we evaluate the cost of the clients in $A \setminus (\widehat{A} \cup C(Z^*))$. 
  Consider such a client $x$ and let $C^*_i$ be the cluster it belongs to in solution $C^*$.
 Since $C^*_i \in C^*\setminus Z^*$,  
  $\local(i)$ is defined and  we have 
  $\local(i) \in \widehat \local \subseteq \calM$.
  Hence, the cost of $x$ in $\calM$ is at most $\cost(x,\local(i))$. Observe that by the triangular inequality,
$\cost(x,\local(i)) \le \cost(x,c^*_i) + \cost(c^*_i,\local(i)) = g_{x} + \cost(c^*_i,\local(i))$.


Now consider a client $x' \in \IR_i^{\varepsilon^2}\cap C^*_i\cap A(\local(i))$. By the triangular inequality, we have
   $\cost(c^*_i,\local(i)) \le \cost(c^*_i,x') + \cost(x',\local(i)) = g_{x'} + l_{x'}$. 
  Hence,
  $$ 
  \cost(c^*_i,\local(i)) \le \frac{1}{| \IR_i^{\varepsilon^2}\cap C^*_i\cap A(\local(i)) |} \sum\limits_{x' \in \IR_i^{\varepsilon^2}\cap C^*_i\cap A(\local(i))} (g_{x'} + l_{x'}).
  $$
  It follows that assigning the clients of $C^*_i \cap (A \setminus \widehat{A})$ to $\local(i)$ induces
  a cost of at most 
  $$
  \sum\limits_{x \in C^*_i \cap (A\setminus \widehat{A})} g_x + \frac{|C^*_i \cap (A\setminus\widehat{A})|}{| \IR_i^{\varepsilon^2}\cap C^*_i\cap A(\local(i))|} 
  \sum\limits_{x' \in  \IR_i^{\varepsilon^2}\cap C^*_i\cap A(\local(i))} (g_{x'} + l_{x'}).
  $$

Due to Lemma~\ref{lem:structbeta}, we have $|\IR_i^{\varepsilon^2}\cap C^*_i\cap A(\local(i))|\geq (1-\varepsilon)\cdot |\IR_i^{\varepsilon^2}\cap C^*_i|$ and $|(\IR_i^{\varepsilon^2}\cap C_i^*)\cap (A\setminus\widehat{A})| \leq \varepsilon\cdot |\IR_i^{\varepsilon^2}\cap C_i^*|$.
Further, $|(C_i^*\setminus \IR_i^{\varepsilon^2})\cap (A\setminus\widehat{A})| \leq |(C_i^*\setminus \IR_i^{\varepsilon^2})| = |C^*_i| - |\IR_i^{\varepsilon^2}\cap C^*_i|$.
Combining these three bounds, we have 
\begin{eqnarray}
\nonumber
\frac{|C_i^*\cap (A\setminus\widehat{A})|}{|\IR_i^{\varepsilon^2}\cap C^*_i\cap A(\local(i))|} &=& \frac{|(C_i^*\setminus \IR_i^{\varepsilon^2})\cap (A\setminus\widehat{A})| + |(C_i^*\cap \IR_i^{\varepsilon^2})\cap (A\setminus\widehat{A})|}{|\IR_i^{\varepsilon^2}\cap C^*_i\cap A(\local(i))|} \\ \nonumber
& \leq &\frac{|C_i^*|-(1-\varepsilon)|\IR_i^{\varepsilon^2}\cap C^*_i|}{(1-\varepsilon)\cdot |\IR_i^{\varepsilon^2}\cap C^*_i|}=  \frac{|C^*_i|}{(1-\varepsilon)\cdot |\IR_i^{\varepsilon^2}\cap C^*_i|} - 1 \\
\label{eq:triangle}
&\le &  \frac{|C_i^*|}{(1-\varepsilon)^2\cdot |C^*_i|} - 1 
 \le  \frac{2\varepsilon -\varepsilon^2}{(1-\varepsilon)^2}
<  \frac{2\varepsilon}{(1-\varepsilon)^2},
\end{eqnarray}
where the inequality in~(\ref{eq:triangle}) follows from Lemma~\ref{lem:sizeIR}.

Summing over all clusters $C^*_i \in C^*\setminus Z^*$, we obtain that the cost in $\calM$ for the clients in $(A\setminus \widehat{A}) \cap C^*_i$ is less than
  $$ \sum\limits_{c \in A\setminus (\widehat{A} \cup C(Z^*))} g_x + \frac{2\varepsilon}{(1-\varepsilon)^2}\cdot (\cost(C^*) + \cost(\local)).
  $$

 By Lemmas \ref{lem:nbin-delta} and \ref{lem:nbout-delta},
  we have $|\calM \setminus \local| + |\local \setminus \calM| = 2\cdot|Z^*| \le (4\eps^{-3 }+ O(\eps^{-2})) \beta^{-1}$.
  By selecting the neighborhood size of Local Search (Algorithm~\ref{alg:LS}) to be greater than this value, we have $(1-\eps/n)\cdot\cost(\local) \le \cost(\calM)$.
  Therefore, combining the above observations, we have
  $$(1-\frac{\eps}{n})\cdot \cost(\local) \le \sum\limits_{x \in \widehat{A}\setminus C(Z^*)} l_x + \sum\limits_{x \in  C(Z^*)} g_x + \sum\limits_{x \in A\setminus (\widehat{A}\cup C(Z^*))}  g_x + \frac{2\eps }{(1-\eps)^2}\cdot 
                                      (\cost(C^*) + \cost(\local)).$$
By simple transformations, we then obtain                                      
  \begin{align*}
   -\frac{\eps}{n} \cdot \cost(\local) + \sum\limits_{x \in A \setminus (\widehat{A})\cup C(Z^*)} l_x  &\le \sum\limits_{x \in A\setminus(\widehat{A})\cup C(Z^*)} g_x + \frac{2\eps }{(1-\eps)^2}\cdot (\cost(C^*)+\cost(\local)).
  \end{align*}
\end{proof}


We now turn to evaluate the cost for the clients that are in  $\widehat A \setminus C(Z^*)$.
For any cluster $C^*_i\in C^*\setminus C(Z^*)$ and for any $x \in C^*_i \setminus A(\local(i))$ 
define $\reas(x)$ to be the cost of $x$ with respect to the center in 
$\local(i)$. 
Note that there exists only one center of $\local$ in $IR^{\varepsilon}_i$ for any cluster $C^*_i\in C^*\setminus C(Z^*)$.
Before going deeper in the analysis, we need the following lemma. 
\begin{lemma}\label{lem:reassign}
  For any $C^*_i\in C^*\setminus C(Z^*)$, we have
  $$ \sum\limits_{x \in C^*_i \setminus A(\local(i))} \reas(x) \le \sum\limits_{x \in C^*_i \setminus A(\local(i))} g_x + \frac{2\varepsilon}{(1-\varepsilon)^2}
  \sum\limits_{x \in C^*_i}(l_x + g_x).
  $$
\end{lemma}

We now partition the clients of cluster $C^*_i \in C^* \setminus Z^*$.
For any $i$, let $B_i$ be the set of clients of $C^*_i$ that 
are served in solution $\local$ 
by a center $\local(j)$ for some $j \neq i$ and $C^*_j \in C^* \setminus Z^*$. 
Moreover, 
let $D_i = 
(A(\local(i)) \cap (\bigcup_{j \neq i} B_j))$.
Finally, define $E_i =(C^*_i \cap \widehat{A})\setminus \bigcup_{j \neq i} D_j$.

\begin{lemma}\label{lem:costgood-delta-part1}
  Let $C^*_i$ be a cluster in $C^* \setminus Z^*$.
  Define the solution $\calM^i = \local \setminus \{ \local(i) \} \cup \{c^*_i\}$ and
  denote by $m^i_x$ the cost of client $x$ in solution $\calM^i$.
Then
  \begin{equation*}
    \sum_{x \in A} m^i_x 
    \le \sum_{\substack{x \in A\setminus \\
    (A(\local(i)) \cup E_i )}} l_x + \sum_{x \in E_i }  g_x  + 
    \sum_{x \in D_i} \reas(x) +
    \sum_{\substack{x \in A(\local(i)) \setminus \\ (E_i \cup D_i)}} l_x + 
    \frac{\eps}{(1-\eps)} (\sum_{x \in E_i} g_x + l_x).
  \end{equation*}
\end{lemma}

We can thus prove the following lemma, which concludes the proof.

\begin{lemma}\label{lem:costgood-delta}
We have
 $$-\eps \cdot \cost(\local) + \sum_{x \in \widehat A \setminus C(Z^*)} l_x \le \sum_{x \in \widehat A  \setminus C(Z^*)} g_x +  \frac{3\eps}{(1-\varepsilon)^2}\cdot
  (\cost(\local) + \cost(C^*)).$$
\end{lemma}

The proof of Theorem \ref{thm:beta-delta} follows from (1) summing the equations
from Lemmas~\ref{lem:costbad-delta} and~\ref{lem:costgood-delta}
and (2) Lemma~\ref{lem:structbeta}. The comparison 
of the structure of the local solution
to the structure of $C^*$ is an immediate corollary of Lemma~\ref{lem:structbeta}.

\section{Perturbation Resilience}
\label{sec:perturbation-resilient}
We first give the definition of $\alpha$-perturbation-resilient instances.

\begin{defn}\label{defn:perturb}
  Let $I = (\clients,\candidates, \cost, k)$ be an instance for the $k$-clustering problem.
  For $\alpha \ge 1$, $I$ is \emph{$\alpha$-perturbation-resilient} if there exists a unique optimal set of
  centers $C^* = \{c^*_1,\ldots,c^*_k\}$ 
  and for any instance $I' = (\clients,\candidates, \cost', k,p)$, such that 
  $$\forall~a,b \in \mathcal{P},~\cost(a,b) \le \cost'(a,b) \le \alpha \cost(a,b),$$
  the unique optimal set of centers is $C^* = \{c^*_1,\ldots,c^*_k\}$.
\end{defn}

For ease of exposition, we assume that $\cost(a,b) = \dist(a,b)$ 
(\ie we work with the $k$-median problem). 
Given solution $S_0$, 
we say that $S_0$ is \emph{$1/\eps$-locally optimal} if any solution $S_1$ such that 
$|S_0 \setminus S_1| + |S_1 \setminus S_0| \le 2/\eps$ has at least $\cost(S_0)$.

\begin{theorem}
\label{thm:perturbation}
  Let $\alpha > 3$. For any instance of the $k$-median problem that is $\alpha$-perturbation-resilient,
  any $2(\alpha-3)^{-1}$-locally optimal solution is the optimal set of centers $\{c^*_1,\ldots,c^*_k\}$.  
\end{theorem}

Moreover, define $l_c$ to be the cost for client $c$ in solution $\local$ and $g_c$ to be its cost in the 
optimal solution $C^*$. Finally, for any sets of centers $S$ and $S_0 \subset S$, define $N_S(S_0)$
to be the set of clients served by a center of $S_0$ in solution $S$, i.e.:
$N_S(S_0) = \{x \mid \exists s \in S_0, \dist(x,s) = \min_{s' \in S} \dist(x,s') \}$.

The proof of Theorem \ref{thm:perturbation} relies on the following theorem of particular interest.

\begin{theorem}[Local-Approximation Theorem.]
  \label{thm:main:perturbation}
  Let $\local$ be a $1/\eps$-locally optimal solution 
  and $\globalS$ be any solution.
  Define $S = \local \cap \globalS$ and 
  $\localt = \local \setminus S$ and $\globalSt = \globalS \setminus S$.
  Then 
  $$\sum_{c \in \Nglob(\globalSt) - \Nloc(\localt)} l_c + 
  \sum\limits_{c \in \Nloc(\localt)} l_c \le 
  \sum_{c \in \Nglob(\globalSt) - \Nloc(\localt)} g_c + 
  (3+2\eps) 
  \sum\limits_{c \in \Nloc(\localt)}  g_c.$$
\end{theorem}



We first show how Theorem \ref{thm:main:perturbation} 
allows us to prove Theorem \ref{thm:perturbation}.
\begin{proof}[Proof of Theorem \ref{thm:perturbation}]
  Given an instance $(\clients, \candidates, \dist,k)$, we define 
  the following instance $I' 
  = (\clients, \candidates, \dist',k)$,
  where $\dist'(a,b)$ is a distance function defined 
  over $\clients \cup \candidates$ that
  we detail below.
  For each client $c \in \Nloc(\localt) \cup \Nglob(\globalSt)$, 
  let $\ell_c$ be the center of $\local$ that 
  serves it in $\local$, for any point $p \neq \ell_c$, we 
  define $\dist'(c,p) = \alpha \dist(c,p)$ and
  $\dist'(c,\ell_c) = \dist(c,\ell_c)$.
  For the other clients we set $\dist' = \dist$.
  Observe that by local optimality, the clustering induced 
  by $\local$ is $\{c^*_1,\ldots,c^*_k\}$ if and
  only if $\local = \globalS$.
  Therefore, the cost of $\globalS$ in instance $I'$ is equal to 
  $$\alpha \sum\limits_{c \in  \Nloc(\localt)} g_c + 
  \sum\limits_{c \in \Nglob(\globalSt) - \Nloc(\localt)} \min(\alpha g_c,l_c)
  + \sum\limits_{c \not\in \Nglob(\globalSt) \cup \Nloc(\localt)} g_c.$$
  On the other hand, the cost of $\local$ in $I'$ is the same as in
  $I$.
  By Theorem \ref{thm:main:perturbation}
  $$\sum\limits_{c \in \Nglob(\globalSt) - \Nloc(\localt)} l_c + 
  \sum\limits_{c \in \Nloc(\localt)} l_c +
  \le \sum\limits_{c \in \Nglob(\globalSt) - \Nloc(\localt)} g_c + 
  (3+\frac{2(\alpha-3)}{2}) \sum\limits_{c \in \Nloc(\localt)}  g_c$$ 
  and by definition of $S$ we have, for each element $c \notin 
  \Nglob(\globalSt) \cup \Nloc(\localt)$, $l_c = g_c$.

  Thus the cost of $\local$ in $I'$ is at most
  $$  (3+\frac{2(\alpha-3)}{2}) \sum\limits_{c \in \Nloc(\localt)}  g_c+
  \sum\limits_{c \in \Nglob(\globalSt) - \Nloc(\localt)} g_c 
  + \sum\limits_{c \not\in \Nglob(\globalSt) \cup \Nloc(\localt)} g_c$$

  Now, observe that for the clients in 
  $\Nglob(\globalSt) - \Nloc(\localt) = \Nglob(\globalSt) \cap \Nloc(S)$,
  we have $l_c \ge g_c$.

  Therefore, we have that the cost of $\local$ is at most the cost of
  $\globalS$ in $I'$ and so by definition of $\alpha$-perturbation-resilience, 
  we have that the clustering $\{c^*_1,\ldots,c^*_k\}$ is the unique
  optimal solution in $I'$.
  Therefore $\local = \globalS$ and the Theorem follows.
\end{proof}

We now turn to the proof of Theorem \ref{thm:main:perturbation}.

Consider the following bipartite graph $\Gamma = (\localt \cup \globalSt, \mathcal{E})$ where
$\mathcal{E}$ is defined as follows. For any center $f \in \globalSt$, we have
$(f,\ell) \in \mathcal{E}$ where $\ell$ is the center of $\localt$ that is the closest
to $f$. Denote $N_{\Gamma}(\ell)$ the neighbors of the point corresponding to 
center $\ell$ in $\Gamma$.

For each edge $(f,\ell) \in \mathcal{E}$,
for any client $c \in \Nglob(f) \setminus \Nloc(\ell)$, we define $\reas_c$ as the cost of reassigning client $c$
to $\ell$. We derive the following lemma.
\begin{lemma}\label{lem:reas}
  For any client $c$,  $\reas_c \le l_c + 2 g_c$.
\end{lemma}
\begin{proof}
  By definition we have $\reas_c = \dist(c, \ell)$. By the triangle inequality
  $\dist(c,\ell) \le \dist(c,f) + \dist(f,\ell)$. Since $f$ serves $c$ in $\globalS$
  we have $\dist(c, f) = g_c$, hence $\dist(c,\ell) \le g_c + \dist(f,\ell)$.
  We now bound $\dist(f,\ell)$. Consider the center $\ell'$ that serves $c$ in solution $\local$.
  By the triangle inequality we have $\dist(f, \ell') \le \dist(f,c) + \dist(c,\ell') = g_c + l_c$.
  Finally, since $\ell$ is the closest center of $f$ in $\local$, we have 
  $\dist(f,\ell) \le \dist(f, \ell') \le g_c + l_c$ and the lemma follows.
\end{proof}

We partition the centers of $\localt$ as follows.
Let $\localt_0$ be the set of centers of $\localt$ that have degree 0 in $\Gamma$.
Let $\localt_{\le \eps^{-1}}$ be the set of centers of $\localt$ that have degree at least one and at most $1/\eps$ in $\Gamma$.
Let $\localt_{>\eps^{-1}}$ be the set of centers of $\localt$ that have degree greater than $1/\eps$ in $\Gamma$.

We now partition the centers of $\localt$ and $\globalSt$ using the neighborhoods of the vertices 
of $\localt$ in $\Gamma$.
We start by iteratively constructing two set of pairs $S_{\le \eps^{-1}}$ and $S_{>\eps^{-1}}$.
For each center $\ell \in \localt_{\le \eps^{-1}} \cup \localt_{>\eps^{-1}}$, we pick a set $A_{\ell}$ of $|\Ngam|-1$ centers 
of $\localt_0$ and define a pair $(\{\ell\} \cup  A_{\ell}, \Ngam)$. We then remove $A_{\ell}$ from $\localt_0$ and repeat.
Let $S_{\le \eps^{-1}}$ be the pairs that contain a center of $\localt_{\le \eps^{-1}}$ and let $S_{> \eps^{-1}}$ be the remaining
pairs.

The following lemma follows from the definition of the pairs.
\begin{lemma}\label{lem:presence}
  Let $(\Rloc,\Rglob)$ be a pair in $S_{\le \eps^{-1}} \cup S_{> \eps^{-1}}$.
  If $\ell \in \Rloc$, then for any $f$ such that 
  $(f,\ell) \in \mathcal{E}$, $f \in \Rglob$.
\end{lemma}

\begin{lemma}
  \label{lem:cost:<=p}
  For any pair $(\Rloc, \Rglob) \in S_{\le \eps^{-1}}$ we have that 
  $$\sum_{c \in \Nglob(\Rglob)
  } l_c \le 
  \sum_{c \in \Nglob(\Rglob)
  } g_c +   
  2 \sum_{\Nloc(\Rloc) } g_c.$$  
\end{lemma}
\begin{proof}
  Consider the mixed solution $M = \local \setminus \Rloc \cup \Rglob$.
  For each point $c$, let $m_c$ denote the cost of $c$ in solution $M$.
  We have the following upper bounds
  $$m_c \le
  \begin{cases}
    g_c &\mbox{if $c \in \Nglob(\Rglob)$.}\\
    \reas_c &\mbox{if $c \in \Nloc(\Rloc) - \Nglob(\Rglob)$ and by Lemma \ref{lem:presence}.}\\
    l_c &\mbox{Otherwise}.
  \end{cases}
  $$
  Now, observe that the solution $M$ differs from $\local$ by at most $2/\eps$ centers.
  Thus, by $1/\eps$-local optimality we have 
  $\cost(\local) \le \cost(M)$. Summing over all clients and simplifying,
  we obtain 
  $$\sum_{c \in \Nglob(\Rglob)} l_c + 
  \sum_{c \in \Nloc(\Rloc) - \Nglob(\Rglob)} l_c \le 
  \sum_{c \in \Nglob(\Rglob) } g_c +
  \sum_{c \in \Nloc(\Rloc) - \Nglob(\Rglob)} \reas_c.$$  
  The lemma follows by combining with Lemma \ref{lem:reas}.
\end{proof}

We now analyze the cost of the clients served by a center of $\local$ 
that has degree greater than $\eps^{-1}$ in $\Gamma$.
The argument is very similar.

\begin{lemma}
 \label{lem:cost:>p}
 For any pair $(\Rloc, \Rglob) \in S_{>\eps^{-1}}$ we have that 
 $$\sum_{c \in \Nglob(\Rglob) 
 } l_c \le \sum_{c \in \Nglob(\Rglob)
 } g_c + 2(1+\eps)\sum_{\Nloc(\Rloc) 
 } g_c.$$  
\end{lemma}
\begin{proof}
  Consider the center $\hat \ell \in \Rloc$ that has in-degree 
  greater than $\eps^{-1}$.
  Let $\hat L = \Rloc \setminus \{\hat \ell\}$.
  For each $\ell \in \hat L$, we associate a center $f(\ell)$ in $\Rglob$ in such
  a way that each $f(\ell) \neq f(\ell')$, for $\ell \neq \ell'$.
  Note that this is possible since $|\hat L| = |\Rglob| - 1$. Let $\tilde f $ be the 
  center of $\Rglob$ that is not associated with any center of $\hat L$.

  Now, for each center $\ell$ of $\hat L$ we consider the mixed solution
  $M^{\ell} = \local - \{\ell\} \cup \{f(\ell)\}$.
  For each client $c$, we bound its cost $m^{\ell}_c$ in solution $M^{\ell}$.
  We have 
  $$m^{\ell}_c =
  \begin{cases}
    g_c &\mbox{if $c \in \Nglob(f(\ell))$.}\\
    \reas_c &\mbox{if $c \in \Nloc(\ell) - \Nglob(f(\ell))$ and by Lemma \ref{lem:presence}.}\\
    l_c &\mbox{Otherwise.}
  \end{cases}
  $$
  Summing over all center $\ell \in \hat L$, 
  we have by $\eps^{-1}$-local optimality
  \begin{align}
    \label{eq:bdcost1}
    \sum_{c \in \Nglob(\Rglob) - \Nglob(\tilde{f})} l_c 
    + \sum_{\ell \in \Rloc} 
    \sum_{c \in \Nloc(\ell)} l_c 
    &\le \nonumber \\ 
    \sum_{c \in \Nglob(\Rglob) - \Nglob(\tilde{f}) } g_c
    &+ \sum_{\ell \in \Rloc} 
    \sum_{c \in \Nloc(\ell)} \reas_c.
  \end{align}
  
  We now complete the proof of the lemma by analyzing the cost of the clients in $\Nglob(\tilde f)$.
  We consider the center $\ell^* \in \hat L$ that minimizes the reassignment cost of its clients.
  Namely, the center $\ell^*$ such that $\sum_{c \in \Nloc(\ell^*)} \reas_c$ is minimized.
  We then consider the solution $M^{(\ell^*,\tilde f)} = \local - \{\ell^*\} \cup \{\tilde f\}$.
  For each client $c$, we bound its cost $m^{(\ell^*,\tilde f)}_c$ in solution $M^{(\ell^*, \tilde f)}$.
  We have 
  $$m^{(\ell^*, \tilde f)}_c \le
  \begin{cases}
    g_c &\mbox{if $c \in \Nglob(\tilde f) $.}\\
    \reas_c &\mbox{if $c \in \Nloc(\ell^*) - \Nglob(\tilde f )$ and by Lemma \ref{lem:presence}.}\\
    l_c &\mbox{Otherwise.}
  \end{cases}
  $$
  Thus, summing over all clients $c$, we have by local optimality
  \begin{equation}
    \label{eq:bdcost2}
    \sum_{c \in \Nglob(\tilde f) } l_c     +
    \sum_{c \in \Nloc(\ell^*) \setminus \Nglob(f(\ell^*) )} l_c 
    \le 
    \sum_{c \in \Nglob(\tilde f) } g_c +
    \sum_{c \in \Nloc(\ell^*) \setminus \Nglob(f(\ell^*) )} \reas_c.
  \end{equation}  
  By Lemma \ref{lem:reas}, combining Equations \ref{eq:bdcost1} and \ref{eq:bdcost2} and averaging over
  all centers of $\hat L$ we have
  $$\sum_{c \in \Nglob(\Rglob) } l_c \le 
  \sum_{c \in \Nglob(\Rglob) } g_c + 2(1+\eps) \sum_{\Nloc(\Rloc) 
  } g_c.$$
\end{proof}

We now turn to the proof of Theorem \ref{thm:main:perturbation}.
\begin{proof}[Proof of Theorem \ref{thm:main:perturbation}]
  Observe first that for any $c \in \Nloc(\localt) \setminus \Nglob(\globalSt)$, we have $l_c \le g_c$.
  This follows from the fact that the center that serves $c$ in $\globalS$ is in $S$ and so in $\local$
  and thus, we have  $l_c \le g_c$.
  Therefore 
  \begin{equation}
    \label{eq:costofNS0}
    \sum_{c \in \Nloc(\localt)  \setminus \Nglob(\globalSt)
    } l_c \le \sum_{c \in \Nloc(\localt)  \setminus \Nglob(\globalSt)
    } g_c.
  \end{equation}

  
  We now sum the equations of Lemmas \ref{lem:cost:<=p} and \ref{lem:cost:>p} over all pairs
  and obtain
  \begin{align*}
    \sum\limits_{(\Rloc,\Rglob)} \sum\limits_{c\in \Nglob(\Rglob) \cup \Nloc(\Rloc)} l_c 
    &\le \sum\limits_{(\Rloc,\Rglob)}\left( \sum\limits_{c\in \Nglob(\Rglob) \cup \Nloc(\Rloc) } g_c   + (2+2\eps)  \sum_{\Nloc(\Rloc)  
      } g_c \right)\\
    \sum\limits_{c \in \Nglob(\globalSt) \cup \Nloc(\localt)} l_c &\le \sum_{c \in \Nglob(\globalSt) \cup \Nloc(\localt)} g_c + 
                                                                                (2+2\eps) \sum_{c \in \Nloc(\localt) 
                                                } g_c.\\
  \end{align*}
  Therefore,
$$ 
\sum\limits_{c \in \Nglob(\globalSt) - \Nloc(\localt)} l_c + \sum\limits_{ \Nloc(\localt)} l_c \le \sum_{c \in \Nglob(\globalSt) - \Nloc(\localt)} g_c + 
(3+2\eps) \sum_{c \in \Nloc(\localt)
} g_c.
$$
\end{proof}

Additionally, we show that the analysis is tight (up to a $(1+\eps)$ factor):
\begin{proposition}
  \label{prop:LB}
For any $\varepsilon>0$, there exists an infinite family of $3-\varepsilon$-perturbation-resilient instances such that for any constant $\eps > 0$,
  there exists a locally optimal solution that has cost at least $3 \opt$.
\end{proposition}
\begin{proof}
Consider a tripartite graph with nodes $O$, $C$, and $L$, where $O$ is the set of optimal centers, $L$ is the set of centers of a locally optimal solution, and $C$ is the set of clients. 
We have $|O|=|L|=k$ and $|C|=k^2$.
We specify the distances as follows.
First, assume some arbitrary but fixed ordering on the elements of $O$, $L$, and $C$.
Then $\dist(O_i)(C_{i,j}) = 1+\eps/3$ and $\dist(L_i)(C_{j,i})=3$ for any $i,j\in [k]$.
All other distances are induced by the shortest path metric along the edges of the graph, i.e. $\dist(O_i)(C_{j,\ell}) = 7+\eps/3$ and $\dist(L_i)(C_{j,\ell}) = 5+2\eps/3$ for $j,\ell\neq i$.
We first note that $O$ is indeed the optimal solution with a cost of $k^2\cdot (1+\varepsilon/3)$.
Multiplying the distances $\dist(O_i)(C_{i,j})$ by a factor of $(3-\varepsilon)$ for all $i\in [k]$ and $j\mod k = i$, still ensures that $O$ is an optimal solution with a cost of $k^2\cdot (1+\varepsilon/3)\cdot (3-\varepsilon)= k^2\cdot 3(1-\varepsilon^2)$, which shows that the instance is $(3-\varepsilon)$-perturbation resilient.

What remains to be shown is that $L$ is locally optimal.
Assume that we swap out $s$ centers. Due to symmetry, we can consider the solution $\{O_i| i \in [s]\}\cup \{L_i | i\in [k]\setminus [s]\}$.
Each of centers $\{O_i| i \in [s]\}$ serve $k$ clients with a cost of $k\cdot s\cdot (1+\varepsilon/3)$.
The remaining clients are served by $\{L_i | i\in [k]\setminus [s]\}$, as $5+2\varepsilon/3<7+\varepsilon/3$.
The cost amounts to $s\cdot(k-s)\cdot 5+2\varepsilon/3$ for the clients that get reassigned and $(k-s)^2\cdot 3$ for the remaining clients. Combining these three figures gives us a cost of $k^2\cdot 3 + ks\varepsilon -s^2\cdot(2+2\varepsilon/3)> k^2\cdot 3+ ks\varepsilon+s^2\cdot 3$.
For $k>\frac{3s}{\varepsilon}$, this is greater than $k^2 3$, the cost of $L$. 
\end{proof}

\section{Spectral Separability}
\label{sec:euclid}

In this section we will study the spectral separability condition for the Euclidean $k$-means problem.

\begin{defn}[Spectral Separation \cite{KuK10}\footnote{The proximity condition of Kumar and Kannan~\cite{KuK10} implies the spectral separation condition.}]
\label{defn:Ssep}
Let $(A,\mathbb{R}^d, ||\cdot||^2, k)$ be an input for $k$-means clustering in Euclidean space and let $\{C^*_1,\ldots C^*_k\}$ denote an optimal clustering of $A$ with centers $S=\{c^*_1,\ldots c^*_k\}$. 
Denote by $C$ an $n\times d$ matrix such that the row $C_{i} = \underset{c_j^*\in S}{\text{argmin }}||A_i-c^*_j||^2$.
Denote by $||\cdot||_2$ the spectral norm of a matrix.
Then $\{C^*_1,\ldots C^*_k\}$ is \emph{$\gamma$-spectrally separated}, if for any pair $(i,j)$ the following condition holds:
\[||c^*_i-c^*_j|| \geq \gamma\cdot\left(\frac{1}{\sqrt{|C^*_i|}} + \frac{1}{\sqrt{|C^*_j|}}\right) ||A-C||_2.\]
\end{defn}

Nowadays, a standard preprocessing step in Euclidean $k$-means clustering is to project onto the subspace spanned by the rank $k$-approximation.
Indeed, this is the first step of the algorithm by Kumar and Kannan~\cite{KuK10} (see Algorithm~\ref{alg:specmeans}).
\begin{algorithm}
  \caption{$k$-means with spectral initialization~\cite{KuK10}}
  \label{alg:kmeanspec}
  \begin{algorithmic}[1]
    \State{Project points onto the best rank $k$ subspace}
    \State{Compute a clustering $C$ with constant approximation factor on the projection}
    \State{Initialize centroids of each cluster of $C$ as centers in the original space} 
    \State{Run Lloyd's $k$-means until convergence}
  \end{algorithmic}
    \label{alg:specmeans}
\end{algorithm}

In general, projecting onto the best rank $k$ subspace and computing a constant approximation on the projection results in a constant approximation
in the original space.
Kumar and Kannan~\cite{KuK10} and later Awasthi and Sheffet~\cite{AwS12} gave tighter bounds if the spectral separation is large enough. 
Our algorithm omits steps 3 and 4.
Instead, we project onto slightly more dimensions and subsequently use Local Search as the constant factor approximation in step 2.
To utilize Local Search, we further require a candidate set of solutions, which is described in Section~\ref{sec:euclidperturb}.
For pseudocode, we refer to Algorithm~\ref{alg:specls}.
Our main result is to show that, given spectral separability, this algorithm is PTAS for $k$-means (Theorem~\ref{thm:spectralptas}).

 \begin{algorithm}
   \caption{SpectralLS}
   \begin{algorithmic}[1]
     \State{Project points $A$ onto the best rank $k/\varepsilon$ subspace}
     \State{Embed points into a random subspace of dimension $O(\varepsilon^{-2}\log n)$}
     \State{Compute candidate centers (Corollary~\ref{cor:discrete})}
     \State{Local Search($\Theta(\varepsilon^{-4})$)}
    	\State{Output clustering}
   \end{algorithmic}
     \label{alg:specls}
 \end{algorithm}

 \begin{theorem}
 \label{thm:spectralptas}
 Let $(A,\mathbb{R}^d,||\cdot||^2,k)$ be an instance of Euclidean $k$-means clustering with optimal clustering $C=\{C_1^*, \ldots C_k^* \}$ and centers $S=\{c_1^*,\ldots c_k^*\}$. 
 If $C$ is more than $3\sqrt{k}$-spectrally separated, then Algorithm~\ref{alg:specls} is a polynomial time approximation scheme.
 \end{theorem}

We first recall the basic notions and definitions for Euclidean $k$-means.
Let $A\in \mathbb{R}^{n\times d}$ be a set of points in $d$-dimensional Euclidean space, where the row $A_i$ contains the coordinates of the $i$th point. 
The singular value decomposition is defined as $A=U\Sigma V^T$, where $U\in \mathbb{R}^{n\times d}$ and $V\in \mathbb{R}^{d\times d}$ are orthogonal and $\Sigma\in \mathbb{R}^{d\times d}$ is a diagonal matrix containing the singular values where per convention the singular values are given in descending order, i.e. $\Sigma_{1,1}=\sigma_1 \geq \Sigma_{2,2}=\sigma_2 \geq \ldots \Sigma_{d,d}=\sigma_{d}$. 
Denote the Euclidean norm of a $d$-dimensional vector $x$ by $||x||=\sqrt{\sum_{i=1}^d x_i^2}$.
The spectral norm and Frobenius norm are defined as $||A||_2 = \sigma_1$ and $||A||_F = \sqrt{\sum_{i=1}^{d} \sigma_i^2}$, respectively. 

The best rank $k$ approximation $\underset{\text{rank}(X)=k}{\min} ||A-X||_F$ is given via $A_k=U_k\Sigma V^T = U\Sigma_k V^T = U\Sigma V_k^T$, where $U_k$, $\Sigma_k$ and $V_k^T$ consist of the first $k$ columns of $U$, $\Sigma$ and $V^T$, respectively, and are zero otherwise.
The best rank $k$ approximation also minimizes the spectral norm, that is $||A-A_k||_2 = \sigma_{k+1}$ is minimal among all matrices of rank $k$.
The following fact is well known throughout $k$-means literature and will be used frequently throughout this section.
\begin{fact}
\label{fact:magicformula}
Let $A$ be a set of points in Euclidean space and denote by $c(A)=\frac{1}{|A|}\sum_{x\in A}x$ the centroid of $A$. Then the $1$-means cost of any candidate center $c$ can be decomposed via
\[\sum_{x\in A}||x-c||^2 = \sum_{x\in A}||x-c(A)||^2 + |A|\cdot ||c(A)-c||^2\]
and
\[\sum_{x\in A}||x-c(A)||^2 = \frac{1}{2 \cdot |A|}\sum_{x\in A}\sum_{y\in A}||x-y||^2.\]
\end{fact}
Note that the centroid is the optimal $1$-means center of $A$. 
For a clustering $C=\{C_1,\ldots C_k\}$ of $A$ with centers  $S=\{c_1,\ldots c_k\}$, the cost is then $\sum_{i=1}^k \sum_{p\in C_i}||p-c_i||^2$.
Further, if $c_i = \frac{1}{|C_i|} \sum_{p\in C_i} p$, we can rewrite the objective function in matrix form by associating the $i$th point with the $i$th row of some matrix $A$ and using the cluster matrix $X\in \mathbb{R}^{n\times k}$ with $X_{i,j} = \begin{cases} \frac{1}{\sqrt{|C_j^*|}} & \text{if } A_i\in C_j^* \\
0 & \text{else} \end{cases}$ to denote membership.
Note that $X^TX = I$, i.e. $X$ is an orthogonal projection and that $||A-XX^TA||_F^2$ is the cost of the optimal $k$-means clustering.
$k$-means is therefore a constrained rank $k$-approximation problem.

We first restate the separation condition.
\begin{defn}[Spectral Separation]
Let $A$ be a set of points and let $\{C_1,\ldots C_k\}$ be a clustering of $A$ with centers $\{c_1,\ldots c_k\}$. Denote by $C$ an $n\times d$ matrix such that $C_{i} = \underset{j\in \{1,\ldots,k\}}{\text{argmin }}||A_i-c_j||^2$. Then $\{C_1,\ldots C_k\}$ is \emph{$\gamma$ spectrally separated}, if for any pair of centers $c_i$ and $c_j$ the following condition holds:
\[||c_i-c_j|| \geq \gamma\cdot\left(\frac{1}{\sqrt{|C_i|}} + \frac{1}{\sqrt{|C_j|}}\right) ||A-C||_2.\]
\end{defn}

The following crucial lemma relates spectral separation and distribution stability.

\begin{lemma}
\label{lem:mainspec}
For a point set $A$, let $C=\{C_1,\ldots ,C_k\}$ be an optimal clustering with centers $S=\{c_1,\ldots ,c_k\}$ associated clustering matrix $X$ that is at least $\gamma\cdot\sqrt{k}$ spectrally separated, where $\gamma>3$. 
For $\varepsilon>0$, let $A_m$ be the best rank $m= k/\varepsilon$ approximation of $A$.
Then there exists a clustering $K=\{C_1',\ldots C_2'\}$ and a set of centers $S_k$, such that
\begin{enumerate}
\item the cost of clustering $A_m$ with centers $S_k$ via the assignment of $K$ is less than $||A_m-XX^TA_m||_F^2$ and
\item $(K,S_k)$ is $\Omega((\gamma-3)^2\cdot \varepsilon)$-distribution stable.
\end{enumerate}
\end{lemma}

We note that this lemma would also allow us to use the PTAS of Awasthi et al.~\cite{ABS10}.
Before giving the proof, we outline how Lemma~\ref{lem:mainspec} helps us prove Theorem~\ref{thm:spectralptas}.
We first notice that if the rank of $A$ is of order $k$, then elementary bounds on matrix norm show that spectral separability implies distribution stability.
We aim to combine this observation with the following theorem due to Cohen et al.~\cite{CfLocalEMMP15}.
Informally, it states that for every rank $k$ approximation, (an in particular for every constrained rank $k$ approximation such as $k$-means clustering), 
projecting to the best rank $k/\varepsilon$ subspace is cost-preserving.

\begin{theorem}[Theorem 7 of \cite{CfLocalEMMP15}]
\label{thm:embedding}
For any $A\in \mathbb{R}^{n\times d}$, let $A'$ be the rank $\lceil k/\varepsilon\rceil$-approximation of $A$. Then there exists some positive number $c$ such that for any rank $k$ orthogonal projection $P$,
\[||A-PA||_F^2 \leq ||A'-PA'||_F^2+ c \leq (1+\varepsilon) ||A-PA||_F^2.\]
\end{theorem}

The combination of the low rank case and this theorem is not trivial as points may be closer to a wrong center after projecting, see also Figure~\ref{fig:spec}.
Lemma~\ref{lem:mainspec} determines the existence of a clustering whose cost for the projected points $A_m$ is at most the cost of $C^*$.
Moreover, this clustering has constant distribution stability as well which, combined with the results from Section~\ref{sec:euclidperturb}, allows us to use Local Search.
Given that we can find a clustering with cost at most $(1+\varepsilon)\cdot ||A_m-XX^TA_m||_F^2$, Theorem~\ref{thm:embedding} implies that we will have a $(1+\varepsilon)^2$-approximation overall.

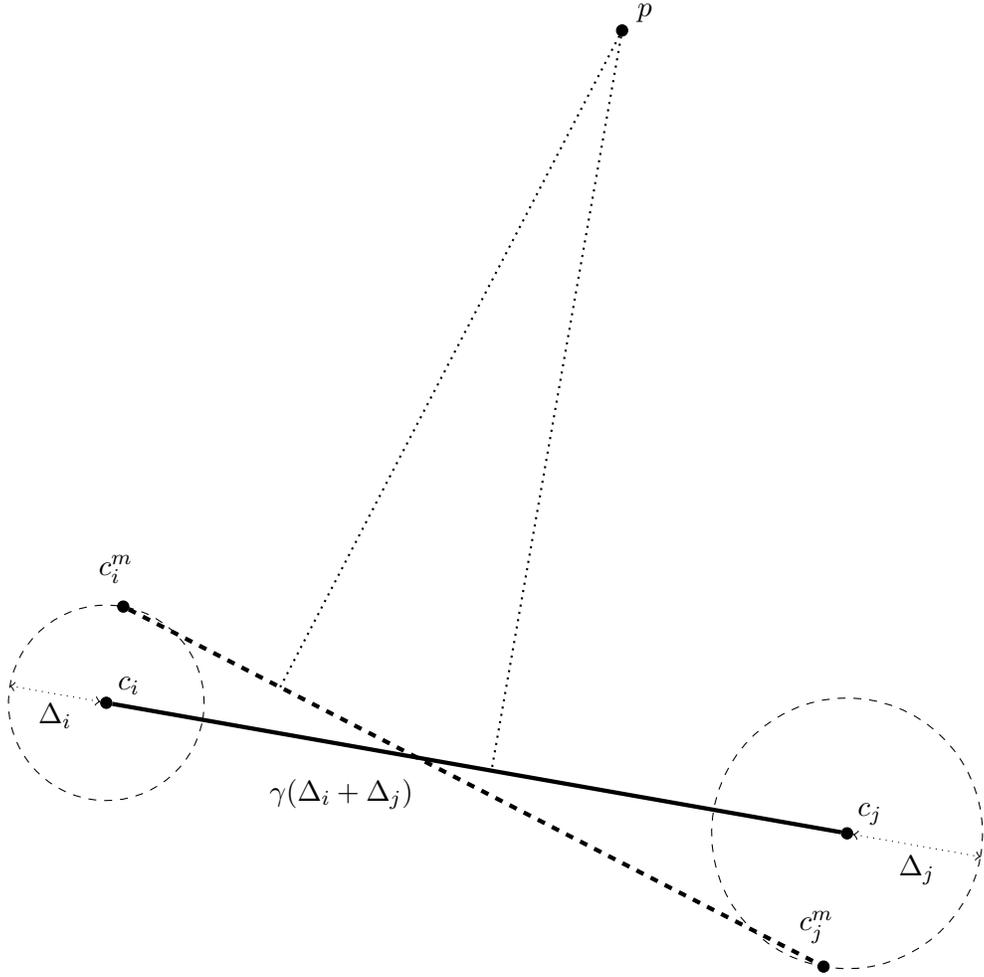
\begin{figure}
\begin{center}
\begin{tikzpicture}[rotate = -10]


\node[draw=black,fill=black, shape=circle, inner sep = 1.5pt, label={[xshift=0.3cm, yshift=-0.1cm]$c_i$}] (c1) at (0,2) {};
\draw[dashed] (0,2) circle (1.3cm);

\node[draw=black,fill=black, shape=circle, inner sep = 1.5pt, label={[xshift=-0.1cm, yshift=0.1cm]$c_i^m$}] (i) at (0,3.3) {};

\node[draw=black,fill=black, shape=circle, inner sep = 1.5pt, label={[xshift=0.3cm, yshift=-0.1cm]$c_j$}] (c2) at (10,2) {};
\draw[dashed] (10,2) circle (1.8cm);
\node[draw=black,fill=black, shape=circle, inner sep = 1.5pt, label={[xshift=-0.1cm, yshift=0.1cm]$c_j^m$}] (j) at (10,0.2) {};

\draw[ultra thick] (c1) --  node[below,label={[xshift=-1.8cm, yshift=-0.7cm]$\gamma(\Delta_i + \Delta_j)$}] {} (c2);
\draw[ultra thick, dashed] (i) --  (j);

\node[draw=black,fill=black, shape=circle, inner sep = 1.5pt, label={[xshift=0.3cm, yshift=-0.1cm]$p$}] (p) at (5.2,12) {};

\draw[thick, dotted] (p) -- (5.2,2);
\draw[thick, dotted] ($(j)!(p)!(i)$) -- (p);

\draw[<->,dotted] (c1) -- node[below] {$\Delta_i$}  (-1.3,2);
\draw[<->,dotted] (c2) -- node[below] {$\Delta_j$}  (10+1.8,2);
\end{tikzpicture}
\end{center}
\caption{Despite the centroids of each cluster being close after computing the best rank $m$ approximation, the projection of a point $p$ to the line connecting the centroid of cluster $C_i$ and $C_j$ can change after computing the best rank $m$ approximation. In this case $||p-c_j||<||p-c_i||$ and $||p-c_i^m|| < ||p-c_j^m||$. (Here $\Delta_i = \sqrt{\frac{k}{|C_i|}}||A-XX^TA||_2$.)}
\label{fig:spec}
\end{figure}

To prove the lemma, we will require the following steps:
\begin{itemize}
\item A lower bound on the distance of the projected centers $||c_iV_mV_m^T - c_jV_mV_m^T|| \approx ||c_i - c_j||$.
\item Find a clustering $K$ with centers $S_m^*=\{c_1V_mV_m^T,\ldots , c_k^*V_mV_m^T\}$ of $A_m$ with cost less than $||A_m-XX^TA_m||_F^2$.
\item Show that in a well-defined sense, $K$ and $C^*$ agree on a large fraction of points.
\item For any point $x\in K_i$, show that the distance of $x$ to any center not associated with $K_i$ is large.
\end{itemize}

We first require a technical statement.


\begin{lemma}
\label{lem:centermove}
For a point set $A$, let $C=\{C_1,\ldots C_k\}$ be a clustering with associated clustering matrix $X$ and let $A'$ and $A''$ be optimal low rank approximations where without loss of generality $k\leq \text{rank}(A') < \text{rank}(A'')$. Then for each cluster $C_i$ 
\[\displaystyle \left\vert\left\vert\frac{1}{|C_i|}\sum_{j\in C_i} \left( A_j''-A_j'\right) \right\vert\right\vert_2 \leq  \sqrt{\frac{k}{|C_i|}}\cdot ||A-XX^TA||_2.\]
\end{lemma}
\begin{proof}
By Fact~\ref{fact:magicformula} $|C_i|\cdot ||\frac{1}{|C_i|}\sum_{j\in C_i} (A_i''-A_i') ||_2^2$ is, for a set of point indexes $C_i$, the cost of moving the centroid of the cluster computed on $A''$ to the centroid of the cluster computed on $A'$.
For a clustering matrix $X$, $||XX^TA'-XX^TA'||_F^2$ is the sum of squared distances of moving the centroids computed on the point set $A''$ to the centroids computed on $A'$. 
We then have \[|C_i|\cdot \left\vert\left\vert\frac{1}{|C_i|}\sum_{j\in C_i} (A_j''-A_j') \right\vert\right\vert_2^2 \leq  ||XX^TA''-XX^TA'||_F^2 \leq  ||X||_F^2 \cdot ||A''-A'||_2^2 \leq k\cdot \sigma_{k+1}^2 \leq k\cdot ||A-XX^TA||_2^2.\] 
\end{proof}


\begin{proof}[Proof of Lemma~\ref{lem:mainspec}]
For any point $p$ associated with some row of $A$, let $p^m = pV_mV_m^T$ be the corresponding row in $A_m$. 
Similarly, for some cluster $C_i$, denote the center in $A$ by $c_i$ and the center in $A_m$ by $c_i^m$.
Extend these notion analogously for projections $p^k$ and $c_i^k$ to the span of the best rank $k$ approximation $A_k$.

We have for any $m\geq k$ $i\neq j$
\begin{eqnarray}
\nonumber
||c_i^m-c_j^m|| &\geq & ||c_i-c_j|| - ||c_i-c_i^m|| - ||c_j-c_j^m|| \\
\nonumber
& \geq & \gamma\cdot \left( \frac{1}{\sqrt{C_i}} + \frac{1}{\sqrt{|C_j|}} \right) \sqrt{k}||A-XX^TA||_2 \\
\nonumber
& &- \frac{1}{\sqrt{C_i}} \sqrt{k}||A-XX^TA||_2 - \frac{1}{\sqrt{|C_j|}} \sqrt{k}||A-XX^TA||_2\\
\label{eq:distcenter}
&=& (\gamma-1)\cdot \left( \frac{1}{\sqrt{C_i}} + \frac{1}{\sqrt{|C_j|}} \right) \sqrt{k}||A-XX^TA||_2,
\end{eqnarray}
where the second inequality follows from Lemma~\ref{lem:centermove}.

In the following, let $\Delta_i = \frac{\sqrt{k}}{\sqrt{|C_i|}}||A-XX^TA||_2$.
We will now construct our target clustering $K$.
Note that we require this clustering (and its properties) only for the analysis.
We distinguish between the following three cases.

\begin{description}
\item[Case 1: $p\in C_i$ and $c_i^m=\underset{j\in\{1,\ldots,k\}}{\text{argmin}}||p^m-c_j||$:] ~

These points remain assigned to $c_i^m$.
The distance between $p_m$ and a different center $c_j^m$ is at least
$\frac{1}{2}||c_i^m-c_j^m||\geq \frac{\gamma-1}{2}\varepsilon (\Delta_i+\Delta_j)$ due to Equation~\ref{eq:distcenter}.

\item[Case 2:  $p\in C_i$, $c_i^m\neq \underset{j\in\{1,\ldots,k\}}{\text{argmin}}||p^m-c_j||$, and $c_i^k \neq \underset{j\in\{1,\ldots,k\}}{\text{argmin}}{||p^k-c_j^k||}$:] ~

These points will get reassigned to their closest center. 

The distance between $p_m$ and a different center $c_j^m$ is at least
$\frac{1}{2}||c_i^m-c_j^m||\geq \frac{\gamma-1}{2}\varepsilon (\Delta_i+\Delta_j)$ due to Equation~\ref{eq:distcenter}.
\item[Case 3: $p\in C_i$, $c_i^m\neq \underset{j\in\{1,\ldots,k\}}{\text{argmin}}{||p^m-c_j^m||}$, and $c_i^k = \underset{j\in\{1,\ldots,k\}}{\text{argmin}}{||p^k-c_j^k||}$:] ~
 
We assign $p^m$ to $c^m_i$ at the cost of a slightly weaker movement bound on the distance between $p^m$ and $c^m_j$.
Due to orthogonality of $V$, we have for $m>k$, $(V_m-V_k)^TV_k=V_k^T(V_m-V_k)=0$. 
Hence $V_mV_m^TV_k = V_mV_k^TV_k + V_m(V_m-V_k)^TV_k = V_kV_k^TV_k + (V_m-V_k)V_k^TV_k = V_kV_k^TV_k = V_k$.
Then $p^k = pV_kV_k^T = pV_mV_m^TV_kV_k^T=p_mV_kV_k^T$.

Further, $||p^k-c_j^k|| \geq \frac{1}{2}||c_j^k - c_i^k|| \geq \frac{\gamma-1}{2}(\Delta_i+\Delta_j)$ due to Equation~\ref{eq:distcenter}.
Then the distance between $p_m$ and a different center $c_j^m$ 
\begin{eqnarray*}
||p^m-c_j^m|| &\geq & ||p^m - c_j^k|| - ||c_j^m - c_j^k|| = \sqrt{||p^m-p^k||^2 + ||p^k - c_j^k||^2} - ||c_j^m - c_j^k|| \\
&\geq & ||p^k-c_j^k|| - \Delta_j \geq \frac{\gamma-3}{2} (\Delta_i+\Delta_j),
\end{eqnarray*} 
where the equality follows from orthogonality and the second to last inequality follows from Lemma~\ref{lem:centermove}.
\end{description}

Now, given the centers $\{c_1^m,\ldots c_k^m\}$, we obtain a center matrix $M_K$ where the $i$th row of $M_K$ is the center according to the assignment of above.
Since both clusterings use the same centers but $K$ improves locally on the assignments, we have $||A_m-M_K||_F^2 \leq ||A_m-XX^TA_m||_F^2$, which proves the first statement of the lemma.
Additionally, due to the fact that $A_m-XX^TA_m$ has rank $m = k/\varepsilon$, we have
\begin{equation}
\label{eq:costbound}
||A_m-M_K||_F^2 \leq ||A_m-XX^TA_m||_F^2 \leq m\cdot ||A_m-XX^TA_m||_2^2 \leq k/\varepsilon \cdot ||A-XX^TA||_F^2
\end{equation}

To ensure stability, we will show that for each element of $K$ there exists an element of $C$, such that both clusters agree on a large fraction of points. 
This can be proven by using techniques from Awasthi and Sheffet~\cite{AwS12} (Theorem 3.1) and Kumar and Kannan~\cite{KuK10} (Theorem 5.4), which we repeat for completeness.

\begin{lemma}
\label{lem:specsize}
Let $K=\{C_1',\ldots C_k'\}$ and $C= \{C_1,\ldots C_k\}$ be defined as above. Then there exists a bijection $b:C\rightarrow K$ such that for any $i\in \{i,\ldots ,k\}$
\[\left( 1-\frac{32}{(\gamma-1)^2}\right) |C_i| \leq b(|C_i|)\leq \left( 1+\frac{32}{(\gamma-1)^2}\right) |C_i|.\]
\end{lemma}
\begin{proof}
Denote by $T_{i\rightarrow j}$ the set of points from $C_i$ such that $||c^k_i-p^k||> ||c^k_j-p^k||$. 
We first note that $||A_k-XX^TA||_F^2 \leq 2k\cdot ||A_k-XX^TA||_2^2 \leq 2k\cdot \left(||A-A_k||_2 + ||A-XX^TA||_2\right)^2 \leq 8k\cdot ||A-XX^TA||_2^2 \leq 8\cdot |C_i| \cdot \Delta_i^2$ for any $i\in\{1,\ldots ,k\}$.
The distance $||p^k-c_i^k||\geq \frac{1}{2}||c^k_i -c_j^k|| \geq \frac{\gamma -1}{2}\cdot\left( \frac{1}{\sqrt{C_i}} + \frac{1}{\sqrt{|C_j|}} \right) \sqrt{k}||A-XX^TA||_2^2$.
Assigning these points to $c^k_i$, we can bound the total number of points added to and subtracted from cluster $C_j$ by observing
\begin{eqnarray*}
\Delta_j^2 \sum_{i\neq j} |T_{i\rightarrow j}| \leq  \sum_{i\neq j} |T_{i\rightarrow j}|\cdot \left(\frac{\gamma-1}{2}\right)^2\cdot (\Delta_i+\Delta_j)^2 \leq ||A_k-XX^TA||_F^2 \leq 8\cdot |C_j|\cdot \Delta_j^2 \\
\Delta_j^2 \sum_{i\neq j} |T_{j\rightarrow i}| \leq  \sum_{j\neq i} |T_{j\rightarrow i}| \cdot \left(\frac{\gamma-1}{2}\right)^2\cdot (\Delta_i+\Delta_j)^2 \leq ||A_k-XX^TA||_F^2 \leq 8\cdot |C_j|\cdot \Delta_j^2.
\end{eqnarray*}
Therefore, the cluster sizes are up to some multiplicative factor of $\left(1\pm \frac{32}{(\gamma-1)^2}\right)$ identical.
\end{proof}

We now have for each point $p^m\in C_i'$ a minimum cost of
\begin{eqnarray*}
||p^m - c_j^m||^2 &\geq & \left(\frac{\gamma-3}{2}\cdot\left(\frac{1}{\sqrt{|C_i|}}+ \frac{1}{\sqrt{|C_j|}}\right)\cdot \sqrt{k}\cdot ||A-XX^TA||_2\right)^2 \\
&\geq & \left(\frac{\gamma-3}{2}\cdot \left(\sqrt{\frac{1}{\left(1+\frac{32}{(\gamma - 1)^2}\right)\cdot |C_i'|}}+ \sqrt{\frac{1}{\left(1+\frac{32}{(\gamma-1)^2}\right)\cdot |C_j'|}}\right)\cdot \sqrt{k}\cdot ||A-XX^TA||_2\right)^2 \\
&\geq & \frac{4\cdot(\gamma-3)^2}{81}\cdot\varepsilon \frac{||A_m-M_K||_F^2}{|C_j'|}
\end{eqnarray*}
where the first inequality holds due to Case 3, the second inequality holds due to Lemma~\ref{lem:specsize} and the last inequality follows from $\gamma>3$ and Equation~\ref{eq:costbound}. 
This ensures that the distribution stability condition is satisfied.
\end{proof}

\begin{proof}[Proof of Theorem~\ref{thm:spectralptas}]
Given the optimal clustering $C^*$ of $A$ with clustering matrix $X$, Lemma~\ref{lem:mainspec} guarantees the existence of a clustering $K$ with center matrix $M_K$ such that $||A_m-M_K||_F^2 \leq ||A_m-XX^TA_m||$ and that $C$ has constant distribution stability.
If $||A_m-M_K||_F^2$ is not a constant factor approximation, we are already done, as Local Search is guaranteed to find a constant factor approximation.
Otherwise due to Corollary~\ref{cor:discrete} (Section~\ref{sec:euclidperturb} in the appendix), there exists a discretization $(A_m,F,||\cdot||^2,k)$ of $(A_m,\mathbb{R}^d,||\cdot||^2,k)$ such that the clustering $C$ of the first instance has at most $(1+\varepsilon)$ times the cost of $C$ in the second instance and such that $C$ has constant distribution stability.
By Theorem~\ref{thm:beta-delta}, Local Search with appropriate (but constant) neighborhood size will find a clustering $C'$ with cost at most $(1+\varepsilon)$ times the cost of $K$ in $(A_m,F,||\cdot||^2,k)$.
Let $Y$ be the clustering matrix of $C'$.
We then have
$||A_m-YY^TA_m||_F^2 +||A-A_m||_F^2 \leq (1+\varepsilon)^2||A_m-M_K||_F^2+||A-A_m||_F^2 \leq (1+\varepsilon)^2||A_m-XX^TA_m||_F^2+||A-A_m||_F^2 \leq (1+\varepsilon)^3||A-XX^TA||_F^2$ due to Theorem~\ref{thm:embedding}.
Rescaling $\varepsilon$ completes the proof.
\end{proof}

\begin{rem}
\label{rem:spec}
Any $(1+\varepsilon)$-approximation will not in general agree with a target clustering. To see this consider two clusters: (1) with mean on the origin and (2) with mean $\delta$ on the the first axis and $0$ on all other coordinates. We generate points via a multivariate Gaussian distribution with an identity covariance matrix centered on the mean of each cluster. If we generate enough points, the instance will have constant spectral separability. However, if  $\delta$ is small and the dimension large enough, an optimal $1$-clustering will approximate the $k$-means objective.
\end{rem}

\section{A Brief Survey on Stability Conditions}
\label{sec:bibstability}
There are two general aims that shape the definitions of stability conditions.
First, we want the objective function to be appropriate. For instance, if the data is generated by mixture of Gaussians, the $k$-means objective will be more appropriate than the $k$-median objective.
Secondly, we assume that there exists some ground truth, i.e. a correct assignment of points into clusters. Our objective is to recover this ground truth as well as possible.
These aims are not mutually exclusive. For instance, an ideal objective function will allow us to recover the ground truth.
We refer to Figure~\ref{fig:tikzmap} for a visual overview of stability conditions and their relationships.

\subsection{Cost-Based Separation}
Given that an algorithm optimized with respect to some objective function, it is natural to define a stability condition as a property the optimum clustering is required to have.

\paragraph{ORSS-Stability~\cite{ORSS12}}
Assume that we want to cluster a data set with respect to the $k$-means objective, but have not decided on the number of clusters.
A simple way of determining the "correct" value of $k$ is to run a $k$-means algorithm for $k\{1,2,\ldots m\}$ until the objective value decreases only marginally (using $m$ centers). At this point, we set $k=m-1$.
The reasoning behind this method, commonly known as the \emph{elbow-method} is that we do not gain much information by using $m$ instead of $m-1$ clusters, so we should favor the simpler model.
Contrariwise, this implies that we did gain information going from $m-2$ to $m-1$ and, in particular, that the $m-2$-means cost was considerably larger than the $m-1$-means cost.

Ostrovsky et al.~\cite{ORSS12} considered whether such discrepancies in the cost also allow us to solve the $k$-means problem more efficiently, see also Schulman~\cite{Sch00} for an earlier condition for two clusters and the irreducibility condition by Kumar et al.~\cite{KSS10}.
Specifically, they assumed that the optimal $k$-means clustering has only an $\varepsilon^2$-fraction of the cost of the optimal $(k-1)$-means clustering.
For such cost separated instances, the popular $D^2$-sampling technique has an improved performance compared to the worst-case $O(\log k)$-approximation ratio~\cite{ArV07,BMORST11,JaG12,ORSS12}.
Awasthi et al.~\cite{ABS10} showed that if an instance is cost-stable, it also admits a PTAS.
In fact, they also showed that the weaker condition \emph{$\beta$-stability} is sufficient. $\beta$-stability states that the cost of assigning a point of cluster $C_i$ to another cluster $C_j$ costs at least $\beta$ times the total cost divided by the size of cluster $C_i$.
Despite its focus on the properties of the optimum,  $\beta$-stability has many connections to target-clustering (see below).
Nowadays, the cost-stable property is one of the strongest stability conditions, implying both distribution stability and
spectral separability (see below).
It is nevertheless the arguably most intuitive stability condition.

\paragraph{Perturbation Resilience}
The other main optimum-based stability condition is ~\emph{perturbation resilience}.
It was originally considered for the weighted max-cut problem by Bilu et al.~\cite{BiL12,BDLS13}. 
There, the optimum max cut is said to be $\alpha$-perturbation resilient, if it remains the optimum even if we multiply any edge weight up to a factor of $\alpha>1$.
This notion naturally extends to metric clustering problems, where, given a $n\times n$ distance matrix, the optimum clustering is $\alpha$-perturbation resilient if it remains optimal if we multiply entries by a factor $\alpha$.
Perturbation resilience has some similarity to smoothed analysis (see Arthur et al.~\cite{AMR11,ArV09} for work on $k$-means).
Both smoothed analysis and perturbation stability aim to study a smaller, more interesting part of the instance space as opposed to worst case analysis that covers the entire space.
Perturbation resilience assumes that the optimum clustering stands out among any alternative clustering and measures the degree by which it stands out via $\alpha$.
Smooth analysis is motivated by considering a problem after applying a random perturbation, which for example accounts for measurement errors.

Perturbation resilience is unique among the considered stability conditions in that we aim to recover the optimum solution, as opposed to finding a good $(1+\varepsilon)$ approximation.
Awasthi et al.~\cite{ABS12} showed that $3$-perturbation resilience is sufficient to find the optimum $k$-median clustering, which was further improved by Balcan and Liang to $1+\sqrt{2}$~\cite{BaL16}~\footnote{These results also holds for a slightly more general condition called the center proximity condition.} and finally to 2 by Angelidakis 
et al.~\cite{MM16}.
Ben-David and Reyzin~\cite{BeR14} showed that recovering the optimal clustering is NP-hard if the instance is less than $2$-perturbation resilient.
Balcan et al.~\cite{BHW16} gave an algorithm that optimally solves symmetric and asymmetric $k$-center on $2$-perturbation resilient instances.
Recently, Angelidakis et al. gave an algorithm that determines the optimum cluster for almost all used center-based clustering if the instance is $2$-perturbation resilient~\cite{MM16}.

\begin{figure*}[h]
\begin{center}
\begin{tikzpicture}
\node[rounded corners=3pt, draw, fill=red!30, text width = 5cm, align = center] (dist) at (4,1) {{\centering\textbf{Distribution Stability}} \\
Awasthi, Blum, Sheffet~\cite{ABS10}};

\node[rounded corners=3pt, draw, fill=red!10, text width = 5cm, align = center] (approx) at (4,-2) {{\centering\textbf{Approximation Stability}} \\
Balcan, Blum, Gupta~\cite{BBG09,BBG13}};

\node[rounded corners=3pt, draw, fill=red!10, text width = 5cm, align = center] (cost) at (4,4) {\textbf{Cost Separation} \\
Ostrovsky, Rabani, Schulman, Swamy~\cite{ORSS12}\\
Jaiswal, Garg~\cite{JaG12}};

\node[rounded corners=3pt, draw, fill=red!30, text width = 5cm, align = center] (spec) at (10,1) {\textbf{Spectral Separation} \\
Kumar, Kannan~\cite{KuK10} \\
Awasthi, Sheffet~\cite{AwS12}};

\node[rounded corners=3pt, draw, fill=red!30, text width = 5cm, align = center] (res) at (-2,1) {\textbf{Perturbation Resilience} \\
Bilu, Daniely, Linial, Saks~\cite{BDLS13,BiL12} \\
Awasthi, Blum, Sheffet~\cite{ABS12} \\
Balcan, Liang~\cite{BaL16}};

\node[rounded corners=3pt, draw, fill=red!10, text width = 5cm, align = center] (prox) at (-2,-2) {\textbf{Center Proximity} \\
Awasthi, Blum, Sheffet~\cite{ABS12}\\
Balcan, Liang~\cite{BaL16}};

\draw[decoration={markings,mark=at position 1 with
    {\arrow[scale=3,>=stealth]{>}}},postaction={decorate}] (approx) -- (dist);
\draw[decoration={markings,mark=at position 1 with
    {\arrow[scale=3,>=stealth]{>}}},postaction={decorate}](cost) -- (dist);
\draw[decoration={markings,mark=at position 1 with
    {\arrow[scale=3,>=stealth]{>}}},postaction={decorate}](cost) to (10,4) to (spec);
\draw[decoration={markings,mark=at position 1 with
    {\arrow[scale=3,>=stealth]{>}}},postaction={decorate}](res) -- (prox);
\end{tikzpicture}
\end{center}
\caption{An overview over all definitions of well-clusterability. Arrows correspond to implication. For example, if an instance is cost-separated then it is distribution-stable; therefore the algorithm by Awasthi, Blum and Sheffet~\cite{ABS10} also works for cost-separated instances. The 
  three highlighted stability definitions in the middle of the figure are considered in this paper.}
\label{fig:tikzmap}
\end{figure*}
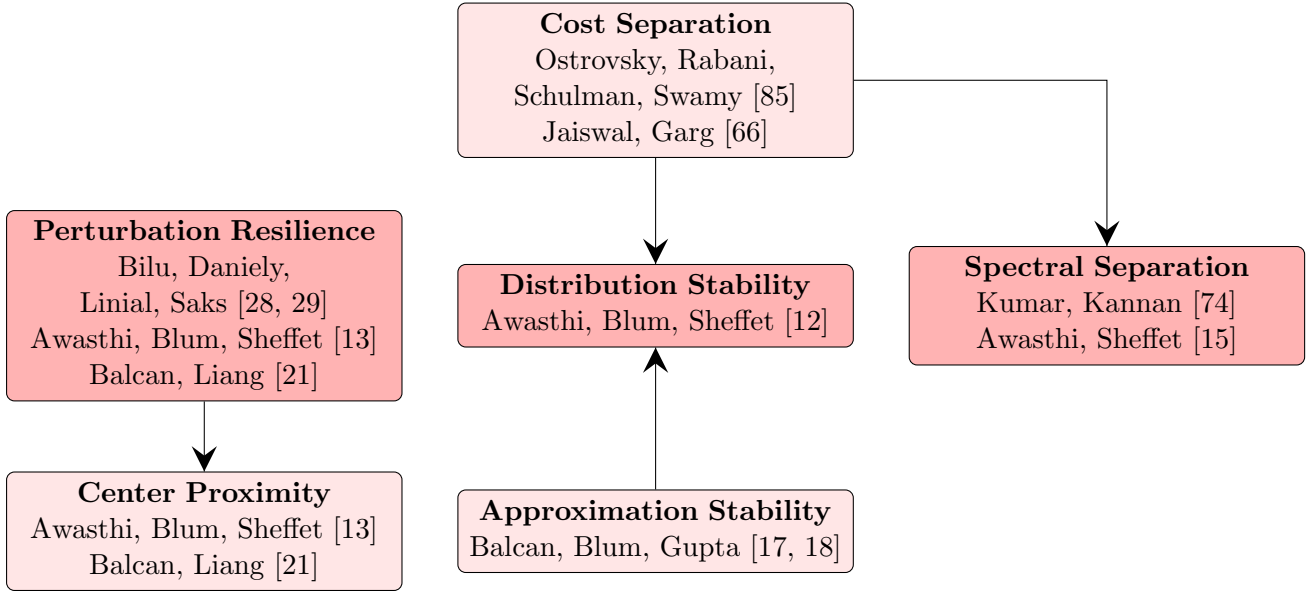

\subsection{Target-Based Stability}
The notion of finding a target clustering is more prevalent in machine learning than minimizing an objective function.
Though optimizing an objective value plays an important part in this line of research, our ultimate goal is to find a clustering $C$ that is close to the target clustering $C^*$. The distance between two clusterings is the fraction of points where $C$ and $C^*$ disagree when considering an optimal matching of clusters in $C$ to clusters in $C^*$.

When the points are generated from some (unknown) mixture model, we are also given an implicit target clustering.
As a result, much work has focused on finding such clusterings using probabilistic assumptions, see, for instance,~\cite{AcM05,ArK01,BeS10,BrV08,Coj10,DHKM07,Das99,DaS07,KSV08,McS01,VeW04}.
We would like to highlight two conditions that make no probabilistic assumptions and have a particular emphasis on the $k$-means and $k$-median objective functions.

\paragraph{Approximation Stability}
The first assumption is that finding the target clustering is related to optimizing the $k$-means objective function.
In the simplest case, the target clustering coincides with the optimum $k$-means clustering, but this a strong assumption that Balcan et al.~\cite{BBG09,BBG13} avoid.
Instead they consider instances where any clustering with cost within a factor $c$ of the optimum has a distance at most $\varepsilon$ to the target clustering, a condition they call $(c,\varepsilon)$-\emph{approximation stability}.
Balcan et al.~\cite{BBG09,BBG13} then showed that this condition is sufficient to both bypass worst-case lower bounds for the approximation factor, and to find a clustering with distance $O(\varepsilon)$ from the target clustering. 
The condition was extended to account for the presence of noisy data by Balcan et al.~\cite{BRT09}.
This approach was improved for other min-sum clustering objectives such as correlation clustering by Balcan and Braverman~\cite{BaB09}.
For constant $c$, $(c,\varepsilon)$ approximation stability also implies the $\beta$-stability condition of Awasthi et al.~\cite{ABS10} with constant $\beta$, if the target clusters are greater than $\varepsilon n$.

\paragraph{Spectral Separability}
Another condition that relates target clustering recovery via the $k$-means objective was introduced by Kumar and Kannan~\cite{KuK10}.
In order to give an intuitive explanation, consider a mixture model consisting of $k$ centers.
If the mixture is in a low-dimensional space, and assuming that we have, for instance, approximation stability with respect to the $k$-means objective, we could simply use the algorithm by Balcan et al.~\cite{BBG13}.
If the mixture has many additional dimensions, the previous conditions have scaling issues, as the $k$-means cost may increase with each dimension, even if many of the additional dimensions mostly contain noise.
The notion behind the \emph{spectral separability} condition is that if the means of the mixture are well-separated in the subspace containing their centers, it should be possible to determine the mixture even with the added noise.

Slightly more formally, Kumar and Kannan state that a point satisfies a proximity condition if the projection of a point onto the line connecting its cluster center to another cluster center is $\Omega(k)$ standard deviations closer to its own center than to the other.
The standard deviations are scaled with respect to the spectral norm of the matrix in which the $i$th row is the difference vector between the $i$th point and its cluster mean.
Given that all but an $\varepsilon$-fraction of points satisfy the proximity condition, Kumar and Kannan~\cite{KuK10} gave an algorithm that computes a clustering with distance $O(\varepsilon)$ to the target.
They also show that their condition is (much) weaker than the cost-stability condition by Ostrovsky et al.~\cite{ORSS12} and discuss some implications of cost-stability on approximation factors.
Awasthi and Sheffet~\cite{AwS12} later showed that $\Omega(\sqrt{k})$ standard deviations are sufficient to recover most of the results by Kumar and Kannan.

\section{Acknowledgments}
The authors thank their dedicated advisor for this project: Claire Mathieu. Without her, this collaboration would not have been possible.

The second author acknowledges the support by Deutsche Forschungsgemeinschaft within the Collaborative Research Center SFB 876, project A2, and the Google Focused Award on Web Algorithmics for Large-scale Data Analysis.

\section*{Appendix}
\appendix
\section{$(\beta,\delta)$-Stability}
\label{appx:beta}

\begin{replemma}{lem:sizeIR}
  Let $C^*_i$ be a cheap cluster. For any $\eps_0$,
we have $|\IR^{\eps_0}_i\cap C_i^*| > (1-{\eps^3}/{\eps_0}) |C^*_i|$.
\end{replemma}
\begin{proof}
  Observe that each client that is not in $IR^{\eps_0}_i$ is at a distance larger than $\eps_0\beta \cost(C^*)/|C^*_i|$ from $c^*_i$.
  Since $C^*_i$ is cheap, the total cost of the clients in $C^*_i=(\IR^{\eps_0}_i\cap C_i^*)\cup (C^*_i \setminus \IR^{\eps_0}_i)$ is at most $\eps^3 \beta \cost(C^*)$ and in particular, the total
  cost of the clients in $C^*_i \setminus \IR^{\eps_0}_i$ does not exceed $\eps^3\beta \cost(C^*)$.
  Therefore, the total number of such clients is at most $\eps^3\beta \cost(C^*)/(\eps_0\beta \cost(C^*)/|C^*_i|)=\eps^3|C^*_i|/\eps_0$.
\end{proof}

\begin{replemma}{lem:IRinter}
  Let $\delta + \frac{\varepsilon^3}{\eps_0} < 1$. If $C^*_i \neq C^*_j$ are cheap clusters, then $\IR^{\eps_0}_i \cap \IR^{\eps_0}_j = \emptyset$.
\end{replemma}
\begin{proof}
  Assume that the claim is not true and consider a client $x \in \IR^{\eps_0}_i \cap \IR^{\eps_0}_j$. Without loss of generality assume $|C^*_i| \ge |C^*_j|$.
  By the triangular inequality, we have
  $\cost(c^*_j, c^*_i) \le \cost(c^*_j, x) + \cost(x, c^*_i) \le \eps_0 \beta \cost(C^*)/|C^*_j|+\eps_0 \beta \cost(C^*)/|C^*_i|\le 2\eps_0 \beta \cost(C^*)/|C^*_j|$.
  Since the instance
  is $(\beta, \delta)$-distribution stable with respect to $(C^*,S^*)$ and due to Lemma~\ref{lem:sizeIR}, we have $|\Delta_i|+ |\IR^{\eps_0}_i\cap C_i^*| > (1-\delta)|C^*_i|+(1-{\eps^3}/{\eps_0}) |C^*_i|= (2-\delta-{\eps^3}/{\eps_0}) |C^*_i|$.  For $\delta + {\eps^3}/{\eps_0} < 1$, there exists a client $x' \in \IR^{\eps_0}_i \cap \Delta_i$. Thus, we have 
  $\cost(x',c^*_j)\le \cost(x',c^*_i) + \cost(c^*_j, c^*_i) \le 3\eps_0 \beta \cost(C^*)/|C^*_j|<\beta \cost(C^*)/|C^*_j|$. Since $x'$ is in $\Delta_i$, we have $\cost(x',c^*_j)\ge \beta\cost(C^*)/|C^*_j|$ resulting in a contradiction.
\end{proof}

\begin{replemma}{lem:nbout-delta}
  There exists a set $Z_2 \subseteq C^* \setminus Z_1$ of size at most 
  $11.25 \eps^{-1}\beta^{-1}$ such that for any cluster
  $C^*_j \in C^* \setminus Z_2$, the total number of clients $x \in 
  \bigcup_{i \neq j} \Delta_i$, that are served by $\local(j)$ in 
  $\local$, is at most $\eps  |\IR^{\eps^2}_i\cap C_i^*|$.
\end{replemma}
\begin{proof}
  Consider a cheap cluster $C^*_j\in C^* \setminus Z_1$
  such that the total number of clients $x \in \Delta_i$ for $i \neq j$, 
  that are served by $\local(j)$ in $\local$, is greater than
  $\eps  |\IR^{\eps^2}_j\cap C^*_j|$.
  By the triangular inequality and the definition of 
  $(\beta,\delta)$-stability, 
  the total cost for each $x \in \Delta_i$ with $i \neq j$ served by 
  $\local(j)$ is at least $(1-\eps)\beta \cost(C^*)/|C^*_j|$.
  Since there are at least $\eps  |\IR^{\eps^2}_j\cap C^*_j|$ such clients, 
  their total cost is
  at least $\eps |\IR^{\eps^2}_j\cap C^*_j| (1-\eps)\beta \cost(C^*)/|C^*_j|$. 
  By Lemma \ref{lem:sizeIR},
  this total cost is at least
  \begin{align*}
\eps |\IR^{\eps^2}_j\cap C^*_j| (1-\eps)\beta \frac{\cost(C^*)}{|C^*_j|} \geq    
    \eps (1-\eps)^2 |C^*_j| \beta \frac{\cost(C^*)}{|C^*_j|}. 
  \end{align*}
  Recall that by \cite{AGKMMP04}, $\local$ is a 5-approximation and so 
  there exist at most $11.25\cdot \eps^{-1}\beta^{-1}$ 
  such clusters. 
\end{proof}

\begin{replemma}{lem:reassign}
  Let $C^*_i$ be a cluster in $C^* \setminus Z^*$.
  Define the solution $\calM^i = \local \setminus \{ \local(i) \} \cup \{c^*_i\}$ and
  denote by $m^i_x$ the cost of client $x$ in solution $\calM^i$.
Then
  \begin{equation*}
    \sum_{x \in A} m^i_x 
    \le \sum_{\substack{x \in A\setminus \\
    (A(\local(i)) \cup E_i )}} l_x + \sum_{x \in E_i }  g_x  + 
    \sum_{x \in D_i} \reas(x) +
    \sum_{\substack{x \in A(\local(i)) \setminus \\ (E_i \cup D_i)}} l_x + 
    \frac{\eps}{(1-\eps)} (\sum_{x \in E_i} g_x + l_x).
  \end{equation*}
\end{replemma}
\begin{proof}
Consider a client $x \in C^*_i \setminus A(\local(i))$.
  By the triangular inequality,
  we have $\reas(x) = \cost(x, \local(i)) \le \cost(x, c^*_i) + \cost(c^*_i, \local(i)) = 
  g_x + \cost(c^*_i, \local(i))$.
  Then,
  $$ 
  \sum\limits_{x \in C^*_i \setminus A(\local(i))} \reas(x) \le \sum\limits_{x \in C^*_i \setminus A(\local(i))} g_x + | C^*_i \setminus A(\local(i))|
  \cdot \cost(c^*_i, \local(i)).
  $$
  
Now consider the clients in $C^*_i\cap A(\local(i))$.
  By the triangular inequality, we have $\cost(c^*_i,\local(i)) \le \cost(c^*_i, x') + \cost(x',\local(i)) \le g_x + l_x$.
  Therefore,
  $$
  \cost(c^*_i,\local(i)) \le \frac{1}{|C^*_i\cap A(\local(i))|} \sum\limits_{x \in C^*_i\cap A(\local(i))} (g_x + l_x).
  $$
  We now bound $\frac{| C^*_i \setminus A(\local(i))|}{|C^*_i\cap A(\local(i))|}$. Due to Lemma~\ref{lem:sizeIR}, we have $|\IR_i^{\varepsilon^2}\cap C^*_i|\geq (1-\varepsilon)|C^*_i|$ and due to Lemma~\ref{lem:structbeta}, we have $|\IR_i^{\varepsilon^2}\cap C^*_i\cap A(\local(i))| \geq (1-\varepsilon)|\IR_i^{\varepsilon^2}\cap C^*_i|$. 
  Therefore $|C^*_i\cap A(\local(i))|\geq (1-\varepsilon)^2|C^*_i|$ and $|C^*_i\setminus A(\local(i))|\leq (1-(1-\varepsilon)^2)|C^*_i|\leq 2\varepsilon|C^*_i|$, yielding $\frac{| C^*_i \setminus A(\local(i))|}{|C^*_i\cap A(\local(i))|}\leq \frac{2\varepsilon}{(1-\varepsilon)^2}$.
  
  Combining, we obtain
  \begin{align*}
  \sum\limits_{x \in C^*_i \setminus A(\local(i))} \reas(x) &\le \sum\limits_{x \in C^*_i \setminus A(\local(i))} g_x + 
  \frac{| C^*_i \setminus A(\local(i))|}{|C^*_i\cap A(\local(i))|}\sum\limits_{x \in C^*_i\cap A(\local(i))} (g_x + l_x)\\
                                                        &\le \sum\limits_{x \in C^*_i \setminus A(\local(i))} g_x +
                                                          \frac{2\varepsilon}{(1-\varepsilon)^2} \sum\limits_{x \in C^*_i\cap A(\local(i))} (g_x + l_x).
  \end{align*}
\end{proof}

\begin{replemma}{lem:costgood-delta-part1}
  Let $C^*_i$ be a cluster in $C^* \setminus Z^*$.
  Define the solution $\calM^i = \local \setminus \{ \local(i) \} \cup \{c^*_i\}$ and
  denote by $m^i_c$ the cost of client $c$ in solution $\calM^i$.
Then
  \begin{equation*}
    \sum_{x \in A} m^i_x 
    \le \sum_{\substack{x \in A\setminus \\
    (A(\local(i)) \cup \tilde C^*_i )}} l_x + \sum_{x \in \tilde C^*_i }  g_x  + 
    \sum_{x \in D_i} \reas(x) +
    \sum_{\substack{x \in A(\local(i)) \setminus \\ (\tilde C^*_i \cup D_i)}} l_x + 
    \frac{\eps}{(1-\eps)} (\sum_{x \in \tilde C^*_i} g_x + l_x).
  \end{equation*}
\end{replemma}
\begin{proof}
   For any client $x \in A\setminus A(\local(i))$,
  the center that serves it in $\local$ belongs to $\calM^i$. 
  Thus its cost is at most $l_x$.
  Moreover, observe that any client 
  $x \in  E_i \subseteq C^*_i$
  can now be served by $c^*_i$, and so its cost is at most $g_x$.
  For each client $x \in D_i$, we bound
  its cost by $\reas(x)$ since all the centers of 
  $\local$ except for $\local(i)$ are in $\calM^i$ and $x\in B^*_j \subseteq C^*_j \in C^*\setminus C(Z^*)$.

 Now, we bound the cost of a client $x \in A(\local(i)) \setminus (E_i \cup D_i) \subseteq A(\local(i))$.
  The closest center in $\calM^i$ for a client 
  $x' \in A(\local(i))$ 
  is not farther than $c^*_i$.
  By the triangular inequality, the cost of such client $x'$ is at 
  most $\cost(x', c^*_i) \le \cost(x',\local(i)) + \cost(\local(i),c^*_i)
  = l_{x'} + \cost(\local(i),c^*_i)$, 
  and so
  \begin{equation}
    \label{eq:cost_analysis_1-delta}
    \sum_{\substack{x \in A(\local(i)) \setminus \\ (E_i \cup D_i)}} m^i_x
    \le   |A(\local(i)) \setminus (E_i \cup D_i)|\cdot
    \cost(\local(i),c^*_i) + 
    \sum_{\substack{x \in A(\local(i)) \setminus \\ (E_i \cup D_i)}} l_x.
  \end{equation}
  Now, observe that, for any client $x \in  
  |A(\local(i))\cap E_i|$, by the triangular inequality, we have 
  $\cost(\local(i),c^*_i) \le \cost(\local(i),x) + \cost(x,c^*_i) 
  = l_x + g_x$. 
  Therefore,
  \begin{equation}
    \label{eq:cost_analysis_2-delta}
    \cost(\local(i),c^*_i) \le \frac{1}{|A(\local(i)) \cap E_i|}
    \sum\limits_{x \in A(\local(i)) \cap E_i} (l_x + g_x).    
  \end{equation}
  Combining Equations \ref{eq:cost_analysis_1-delta} and \ref{eq:cost_analysis_2-delta}, we have  
  \begin{eqnarray}
  \nonumber
    \sum_{\substack{ x \in A(\local(i)) \setminus \\ (E_i \cup D_i)}} m^i_x &\le & \sum_{\substack{ x \in A(\local(i)) \setminus \\ (E_i \cup D_i)}} l_x + 
    \frac{|A(\local(i)) \setminus (E_i \cup D_i)|}{|A(\local(i)) \cap E_i|}
    \sum\limits_{\substack{ x \in A(\local(i)) \cap E_i}} (l_x + g_x)\\
        \label{eq:cost_analysis_3-delta}
        &\leq &
 \sum_{\substack{ x \in A(\local(i)) \\ \setminus (E_i \cup D_i)}} l_x + 
    \frac{|A(\local(i)) \setminus E_i|}{|A(\local(i)) \cap E_i|}
    \sum\limits_{\substack{ x \in E_i}} (l_x + g_x).
  \end{eqnarray}
  We now remark that since $E_i$ is in $C^*\setminus Z^*$, 
  we have by Lemmas \ref{lem:nbin-delta} and \ref{lem:nbout-delta},
  $|A(\local(i)) \setminus E_i| \le \eps \cdot|IR^{\eps^2}_i\cap C^*_i|$ and $(1-\eps)\cdot |IR^{\eps^2}_i\cap C^*_i| \le |A(\local(i))\cap E_i|$.
  Thus, combining with Equation \ref{eq:cost_analysis_3-delta} yields the lemma.  
\end{proof}

\begin{replemma}{lem:costgood-delta}
We have
 $$-\eps \cdot \cost(\local) + \sum_{x \in \widehat A \setminus C(Z^*)} l_x \le \sum_{x \in \widehat A  \setminus C(Z^*)} g_x +  \frac{3\eps}{(1-\varepsilon)^2}\cdot
  (\cost(\local) + \cost(C^*)).$$
\end{replemma}
\begin{proof}
  We consider a cluster $C^*_i$ in $C^* \setminus Z^*$
  and the solution $\calM^i = \local \setminus \{ \local(i) \} \cup \{c^*_i\}$.
  Observe that $\calM^i$ and $\local$ only differ by $\local(i)$ and $c^*_i$.
  Therefore, by local optimality 
  we have $(1-\frac{\eps}{n})\cdot \cost(\local_i) \le \cost(\calM^i)$.
  Then Lemma \ref{lem:costgood-delta-part1} yields
  $$
  (1-\frac{\eps}{n})\cdot \cost(\local_i) \le  \sum_{\substack{x \in A\setminus \\
    (A(\local(i)) \cup E_i )}} l_x + \sum_{x \in E_i } g_x + \sum_{x \in D_i} \reas(x)
  + 
  \sum_{\substack { x \in A(\local(i)) \setminus \\ (E_i \cup D_i)}} l_x 
  +  \frac{\eps}{(1-\eps)}\cdot  \sum_{x \in E} (g_x + l_x)
  $$
  and so, simplifying
  $$
  -\frac{\eps}{n}\cdot \cost(\local_i) + \sum_{x \in E_i} l_x + \sum_{x \in D_i} l_x \le  \sum_{x \in E_i} g_x
  + \sum_{x \in D_i} \reas(x) + 
  \frac{\eps}{(1-\eps)} \cdot \sum_{x \in E_i} (g_x + l_x)
  $$
  We now apply this analysis to each cluster $ C^*_i \in C^* \setminus Z^*$.
  Summing over all clusters $C^*_i$, we obtain,
  \begin{align*}
    -\frac{\eps}{n} \cdot \cost(\local) + &\sum_{i=1}^{|C^* \setminus Z^*|} \left( \sum_{x \in E_i } l_x 
    + \sum_{x \in D_i} l_x \right) 
    \le \\
                                     &\sum_{i=1}^{|C^* \setminus Z^*|} \left(
    \sum_{x \in E_i} g_x + \sum_{x \in D_i} \reas(c) \right)
                                       + \frac{\eps}{(1-\eps)}\cdot (\cost(\local) + \cost(C^*))
  \end{align*}
  By Lemma \ref{lem:reassign} and the definition of $E_i$,
  \begin{eqnarray*}
    -\frac{\eps}{n}\cdot  \cost(\local) + \sum_{i=1}^{|C^* \setminus Z^*|} \sum_{x \in C^*_i \cap \widehat A} l_x \\
    \le \sum_{i=1}^{|C^* \setminus Z^*|}
    \sum_{x \in C^*_i\cap \widehat A} g_x + \left(\frac{\eps}{1-\varepsilon} + \frac{2\eps}{(1-\eps)^2}\right)\cdot (\cost(\local) + \cost(C^*)).
  \end{eqnarray*}
  Therefore, $\displaystyle -\frac{\eps}{n} \cdot \cost(\local) + \sum_{x \in \widehat A  \setminus C(Z^*)} l_x \le \sum_{x \in \widehat A \setminus C(Z^*)} g_x +  
  \frac{3\eps}{(1-\eps)^2} \cdot (\cost(\local) + \cost(C^*)).$
\end{proof}

\section{Euclidean Distribution Stability}
\label{sec:euclidperturb}

In this section we show how to reduce the Euclidean problem to the discrete version.
Our analysis is focused on the $k$-means problem, however we note that the discretization works for all values of $\cost = \dist^p$, where the dependency on $p$ grows exponentially.
For constant $p$, we obtain polynomial sized candidate solution sets in polynomial time.
For $k$-means itself, we could alternatively combine Matousek's approximate centroid set~\cite{Mat00} with the Johnson Lindenstrauss lemma and avoid the following construction; however this would only work for optimal distribution stable clusterings and the proof Theorem~\ref{thm:spectralptas} requires it to hold for non-optimal clusterings as well.

First, we describe a discretization procedure.
It will be important to us that the candidate solution preserves (1) the cost of any given set of centers and (2) distribution stability.

For a set of points $P$, a set of points $\mathcal{N}_{\varepsilon}$ is an \emph{$\varepsilon$-net} of $P$ if for every point $x\in P$ there exists some point $y\in \mathcal{N}_{\varepsilon}$ with $||x-y||\leq \varepsilon$.
It is well known that for unit Euclidean ball of dimension $d$, there exists an $\varepsilon$-net of cardinality $(1+2/\varepsilon)^d$, see for instance Pisier~\cite{Pis99}, though in this case the proof is non-constructive.
Constructive methods yield slightly worse, but asymptotically similar bounds of the form $\varepsilon^{-O(d)}$, see for instance Chazelle~\cite{Cha01} for an extensive overview on how to construct such nets. Note that having constructed an $\varepsilon$-net for the unit sphere, we also have an $\varepsilon\cdot r$-net for any sphere with radius $r$.
The following lemma shows that a sufficiently small $\varepsilon$-net preserves distribution stability.
Again for ease of exposition, we only give the proof for $p=1$, and assuming we can construct an appropriate $\varepsilon$-net, but similar results also hold for $(k,p)$ clustering as long as $p$ is constant.

\begin{lemma}
\label{lem:discrete}
Let $A$ be a set of $n$ points in $d$-dimensional Euclidean space and let $\beta,\varepsilon>0$ with $\min(\beta,\varepsilon)>2\eta>0$  be constants.
Suppose there exists a clustering $C=\{C_1,\ldots ,C_k\}$ with centers $S=\{c_1,\ldots c_k\}$ such that
\begin{enumerate}
\item $\text{cost}(C,S) = \sum_{i=1}^k\sum_{x\in C_i} ||x-c_i||$ is a constant approximation to the optimum clustering and \item $C$ is $\beta$-distribution stable. 
\end{enumerate} 
Then there exists a discretization $D$ of the solution space such that there exists a subset $S'=\{c_1',\ldots c_k'\}\subset D$ of size $k$ with
\begin{enumerate}
\item $\sum_{i=1}^k \sum_{x\in C_i} ||x-c_i'|| \leq (1+\varepsilon)\cdot \text{cost}(C,S)$ and
\item $C$ with centers $S'$ is $\beta/2$-distribution stable.
\end{enumerate}
The discretization consists of $O(n\cdot \log n\cdot \eta^{d+2})$ many points.
\end{lemma}
\begin{proof}
Let $\opt$ being the cost of an optimal $k$-median clustering.
Define an exponential sequence to the base of $(1+\eta)$ starting at $(\eta\cdot \frac{\opt}{n})$ and ending at $(n\cdot \opt)$. 
The sequence contains $t=\log_{1+\eta} (n^2/\eta) \in O(\eta^{-1}\log n)$ many elements for $1/\eta<n$.
For each point $p\in A$, define $B(p,\ell_i)$ as the $d$-dimensional ball centered at $p$ with radius $(1+\eta)^i\cdot \eta\cdot \frac{\opt}{n}$.
We cover the ball $B(p,\ell_i)$ with an $\eta/8 \cdot\ell_i$ net denoted by $\mathcal{N}_{\eta/8}(p,\ell_i)$. 
As the set of candidate centers, we let $D=\cup_{p\in A}\cup_{i=0}^t \mathcal{N}_{\eta/8}(p,\ell_i)$.
Clearly, $|D|\in O(n\cdot\log n\cdot (1+16/\eta)^{d+2})$.

Now for each $c_i\in S$, set $c_i'=\underset{q\in D}{\text{argmin }}||q-c_i||$.
We will show that $S'=\{c_1',\ldots c_k'\}$ satisfies the two conditions of the lemma.

For (1), we first consider the points $p$ with $||p-c_i|| \leq \varepsilon/8\cdot \frac{\opt}{n}$. 
Then there exists a $c_i'$ such that $||p-c_i'|| \leq (\eta/8+ \varepsilon/8) \frac{\opt}{n} \leq \varepsilon/4 \frac{\opt}{n}$ and summing up over all such points, we have a total contribution to the objective value of at most $\varepsilon/4\cdot \opt$.

Now consider the remaining points. Since the $\text{cost}(C,S)$ is a constant approximation, the center $c_i$ of each point $p$ satisfies $(1+\eta)^i \cdot \eta\cdot \frac{\opt}{n} \leq ||c_i-p|| \leq  (1+\eta)^{i+1} \cdot \eta\cdot \frac{\opt}{n}$ for some $i\in \{0,\ldots t\}$.
Then there exists some point $q\in \mathcal{N}_{\eta/8}(p,\ell_{i+1})$ with $||q-c_i||\leq \eta/8\cdot (1+\eta)^{i+1} \cdot \eta\cdot \frac{\opt}{n} \leq \eta/8\cdot (1+\eta) ||p-c_i|| \leq \eta/4 ||p-c_i||.$
We then have $||p-c_i'|| \leq (1+\eta/4)||p-c_i||$.
Summing up over both cases, we have a total cost of at most $\varepsilon/4\cdot \opt + (1+\eta/4)\cdot \text{cost}(C,S') \leq (1+\varepsilon/2)\cdot \text{cost}(C,S')$.

To show (2), let us consider some point $p\notin C_j$ with $||p-c_j|| > \beta \cdot \frac{\opt}{|C_j|}$. 
Since $\beta \cdot \frac{\opt}{|C_j|}\geq 2\eta \cdot \frac{\opt}{n}$, there exists a point $q$ and an $i\in \{0,\ldots t\}$ such that $\beta/8 \cdot(1+\eta)^i\cdot \frac{\opt}{n} \leq ||c_i-q|| \leq  \beta/8 \cdot(1+\eta)^{i+1} \cdot \frac{\opt}{n}$.
Then $||c_j'-c_j||\leq  \beta\cdot (1+\eta)^{i+1} \cdot \frac{\opt}{n}$.
Similarly to above, the point $c_j'$ satisfies $||p-c_j'|| \geq ||p-c_j||-||c_j-c_j'|| \geq \beta\cdot \frac{\opt}{|C_j|} - \beta/8(1+\eta) \cdot \frac{\opt}{n}  \geq (1-1/4)\beta\cdot \frac{\opt}{|C_j|} > \beta/2\cdot \frac{\opt}{|C_j|}$.
\end{proof}

To reduce the dependency on the dimension, we combine this statement with the seminal theorem originally due to Johnson and Lindenstrauss~\cite{JoL84}.

\begin{lemma}[Johnson-Lindenstrauss lemma]
\label{lem:jllemma}
For any set of $n$ points $N$ in $d$-dimensional Euclidean space and any $0<\varepsilon<1/2$, there exists a distribution $\mathcal{F}$ over linear maps $f:\ell_2^d \rightarrow \ell_2^m$ with $m\in O(\varepsilon^{-2}\log n)$ such that
\[\mathbb{P}_{f\sim \mathcal{F}}[\forall x,y\in N,~ (1-\varepsilon)||x-y||\leq ||f(x)-f(y)|| \leq (1+\varepsilon) ||x-y||] \geq \frac{2}{3}.\]
\end{lemma}

It is easy to see  that Johnson-Lindenstrauss type embeddings preserve the Euclidean $k$-means cost of any clustering, as the cost of any clustering can be written in terms of pairwise distances (see also  Fact~\ref{fact:magicformula} in Section~\ref{sec:euclid}).
Since the distribution over linear maps $\mathcal{F}$ can be chosen obliviously with respect to the points, this extends to distribution stability of a set of $k$ candidate centers as well.

Combining Lemmas~\ref{lem:jllemma} and~\ref{lem:discrete} gives us the following corollary.

\begin{cor}
\label{cor:discrete}
Let $A$ be a set of points in $d$-dimensional Euclidean space with a clustering $C=\{C_1,\ldots C_k\}$ and centers $S=\{c_1,\ldots c_k\}$ such that $C$ is $\beta$-perturbation stable. 
Then there exists a $(A,F,||\cdot||^2,k)$-clustering instance with clients $A$, $n^{\text{poly}(\varepsilon^{-1}})$ centers $F$ and a subset $S'\subset F\cup A$ of $k$ centers such that $C$ and $S'$ is $O(\beta)$ stable and the cost of clustering $A$ with $S'$ is at most $(1+\varepsilon)$ times the cost of clustering $A$ with $S$.
\end{cor}

\begin{rem}
This procedure can be adapted to work for general powers of cost functions. 
For Lemma~\ref{lem:discrete}, we simply rescale $\eta$.
The Johnson-Lindenstrauss lemma can also be applied in these settings, at a slightly worse target dimension of $O((p+1)^2\log((p+1)/\varepsilon)\varepsilon^{-3}\log n)$, see Kerber and Raghvendra~\cite{KeR15}.
\end{rem}

\section{Experimental Results}
\label{S:experiments}
In this section, we discuss the empirical applicability of
stability as a model to capture real-world data. 
Theorem~\ref{thm:beta-delta} states
that local search with neighborhood of size $n^{\Omega(\eps^{-3}\beta^{-1})}$
returns a solution of cost at most $(1+\eps)\opt$.
Thus, we ask the following question.
\MyFrame{For which values of $\beta$ are the random and real instances
  $\beta$-distribution-stable?
}

We focus on the $k$-means objective 
and we consider real-world and random instances with
ground truth clustering and study under which conditions
the value of the solution induced by the ground truth clustering
is close to the value of the optimal clustering with respect to 
the $k$-means objective.
Our aim is to determine (a range of) values of $\beta$ for which various data sets satisfy distribution stability.

\subsection*{Setup}
The machines used for the experiments have a processor
\texttt{Intel(R) Core(TM) i7\-3770 CPU, 3.40GHz} with four cores and a total virtual memory of 8GB running on an Ubuntu 12.04.5 LTS operating system.
We implemented the Algorithms in C++ and Python. The C++ compiler
is \texttt{g++ 4.6.3}.
Our experiments always used Local Search with a neighborhood of size $1$.
At each step, the neighborhood of the current solution was explored 
in parallel: 8 threads were created by a Python script and each of 
them correspond to a C++ subprocess that explores a 1/8 fraction of 
the space of the neighboring solutions.
The best neighboring solution found by the 8 threads was taken for the 
next step.
For Lloyd's algorithm we use the C++ implementation by Kanungo et al.~\cite{KMNPSW04} available online.

To determine the stability parameter $\beta$, we also required a lower bound on the cost. This was done via a linear relaxation describe in Algorithm~\ref{alg:LPkmeans}.
The LP for the linear program was generated via a Python script and solved using the solver \texttt{CPLEX}.
The average ratio between our upper bound given via Local Search and lower bounds given via Algorithm~\ref{alg:LPkmeans} is 1.15 and the variance for the value 
of the optimal fractional solution 
that is less than $0.5\%$ of the value of the optimal solution. 
Therefore, our estimate of $\beta$ is quite accurate.

\begin{algorithm}
  \textbf{Input:} A set of clients $\clients$, a set of candidates 
  centers $\candidates$, a number of centers $k$, a distance function 
  $\dist$.
  $$\min \sum_{a \in \clients} \sum_{b \in \candidates} x_{a,b}\cdot 
  \dist(a,b)^2
  $$ 
  subject to,
  $$
  \begin{array}{l@{\quad} r c c c}
    &\sum_{b \in \candidates} y_b &\le& k\\ 
    \forall a \in \clients,& \sum_{b \in \candidates} x_{a,b} &=& 1\\
    \forall a \in \clients,~\forall b \in \candidates,&
    y_b &\ge& x_{a,b}\\
    \forall a \in \clients,~\forall b \in \candidates,& x_{a,b}
                                  &\ge&0
  \end{array}
  $$
  \caption{Linear relaxation for the $k$-means problem.}
  \label{alg:LPkmeans}
\end{algorithm}



\subsection{Real Data}
In this section, we focus on four classic real-world datasets with
ground truth clustering: \texttt{abalone}, 
\texttt{digits}, \texttt{iris}, and \texttt{mo\-ve\-ment\_libras}.
\texttt{abalone}, \texttt{iris}, and \texttt{mo\-ve\-ment\_libras}
have been used in various works 
(see~\cite{DMRBP09,DH73,FSBG02,F36,SRG04}
for example) and are available online at the UCI Machine learning 
repository~\cite{L13}.

The \texttt{abalone} dataset consists of 8 physical characteristics 
of all the individuals of a population of abalones. 
Each abalone corresponds to a point in a 
8-dimensional Euclidean space.
The ground truth clustering consists in 
partitioning the points according to the age of the abalones.

The \texttt{digits} dataset consists of 8px-by-8px images of handwritten
digits from the standard machine learning library 
scikit-learn~\cite{Petal11}. Each image is associated to a point
in a 64-dimensional Euclidean space where each pixel corresponds to a 
coordinate. The ground truth clustering consists in partitioning 
the points according to the number depicted in their corresponding
images.

The \texttt{iris} dataset consists of the sepal and petal lengths 
and widths of all the individuals of a population of iris plant 
containing 3 different types of iris plant. Each plant is associated
to a point in 4-dimensional Euclidean space. The ground truth
clustering consists in partitioning the points according to the type of 
iris plant of the corresponding individual.

The \texttt{Movement\_libras} dataset consists of a set of instances
of 15 hand movements in LIBRAS\footnote{LIBRAS is the official Brazilian
sign language}. Each instance is a curve that is mapped in a 
representation with 90 numeric values representing the coordinates 
of the movements. 
The ground truth clustering consists in partitioning the 
points according to the type of the movement they correspond to.

\begin{table}[h]
  \centering
  \footnotesize
  \begin{tabular}{|l|r|r|r|r|}
    \hline
    Properties & \texttt{Abalone} & \texttt{Digits} & \texttt{Iris} & \texttt{Movement\_libras}\\
      \hline
      Number of points  & 636 & 1000 & 150 & 360\\
      \hline
      Number of clusters & 28 & 10 & 3 & 15 \\
      \hline
      Value of ground truth clustering & 169.19 & 938817.0 & 96.1 & 780.96\\ 
      \hline
      Value of fractional relaxation &4.47& 855567.0&83.96 & 366.34\\
      \hline
      Value of Algorithm~\ref{alg:LS} &4.53& 855567.0&83.96 &369.65\\ 
      \hline
      $\%$ of pts correct. class. by Alg.~\ref{alg:LS}
       & 17& 76.2 & 90 & 39\\
      \hline
      $\beta$-stability & 1.27e-06& 0.0676&0.2185 & 0.0065\\ 
      \hline
    \end{tabular}
    \caption{Properties of the real-world instances with ground truth clustering. The neighborhood size for Algorithm~\ref{alg:LS} is 1.}
    \label{Table:realinst}
  \end{table}

Table~\ref{Table:realinst} shows the properties of the four instances.

For the \texttt{Abalone} and \texttt{Movement\_libras} instances,
the values of an optimal solution is much smaller than the value
of the ground truth clustering. 
Therefore the $k$-means objective function might not be ideal as a recovery mechanism. 
Since Local Search optimizes with respect to the $k$-means objective, the clustering output by Local Search is far from the ground truth clustering for those instances:
the percentage of points correctly classified by 
Algorithm~\ref{alg:LS} is at most $17\%$ for the 
\texttt{Abalone} instance and at most $39\%$ for the 
\texttt{Movement\_libras} instance.
For the \texttt{Digits} and \texttt{Iris} instances
the value of the ground truth clustering is at most 
1.15 times the optimal value.
In those cases, the number of points correctly classified 
is much higher: $90\%$ for the \texttt{Iris} instance and
$76.2\%$ for the \texttt{Digits} instance.

The experiments also show that the $\beta$-distribution-stability 
condition is satisfied for $\beta > 0.06$ for the \texttt{Digits},
\texttt{Iris} and \texttt{Movement\_libras} instances. This shows 
that the $\beta$-distribution-stability
condition captures the structure of some famous real-world instances
for which the $k$-means objective is meaningful for finding
the optimal clusters. We thus make the following observations.
\begin{obs}\label{obs:realworld:beta}
  If the value of the ground truth clustering is close
  to the value of the optimal solution, then one can expect
  the instance satisfy the $\beta$-distribution stability 
  property for some  constant $\beta$.
\end{obs}

The experiments show that Algorithm~\ref{alg:LS} with
neighborhood size 1 ($s=1$) is very efficient for all those instances
since it returns a solution whose value is within 2\% of the optimal 
solution for the \texttt{Abalone} instance and a within 0.002\% for 
the other instances. 
Note that the running time of Algorithm~\ref{alg:LS} with $s =1$ 
is $\tilde O(k\cdot n/\eps)$ (using a set of $O(n)$ candidate centers)
and less than 15 minutes for all the instances.
We make the following observation.
\begin{obs}\label{obs:realworld:LS}
  If the value of the ground truth clustering is close 
  to the value of the optimal solution, 
  then one can expect both clusterings to agree 
  on a large fraction of points.
\end{obs}

Finally, observe that for those instances the value of an optimal 
solution to the fractional relaxation of
the linear program is very close to the optimal value of an optimal 
integral solution (since the cost of the integral solution is
smaller than the cost returned by Algorithm~\ref{alg:LS}).
This suggests that the fractional relaxation 
(Algorithm~\ref{alg:LPkmeans})
might have a small integrality gap for real-world instances.
\paragraph{Open Problem:} We believe that it would be interesting
to study the integrality gap of the classic LP relaxation for the 
$k$-median and $k$-means problems under the stability assumption
(for example $\beta$-distribution stability).

\subsection{Data generated from a mixture of $k$ Gaussians}
The synthetic data was generated via a Python script using \texttt{numpy}.
The instances consist of 1000 points generated from a mixture of $k$ Gaussians with 
the same variance $\sigma$ lying in $d$-dimensional space, where $d \in \{5,10,50\}$ and $k \in \{5,50,100\}$.
We generate 100 instances for all possible combinations of the parameters.
The means of the $k$ Gaussians are chosen uniformly and independently at
random in $\Q^d \cap (0,1)^d$. 
The ground truth clustering is the family of sets
of points generated by the same Gaussian.
We compare the value of the ground truth clustering
to the optimal value clustering.


\begin{figure}[h]
  \begin{subfigure}[t]{0.5\textwidth}
    \includegraphics[scale=0.4]{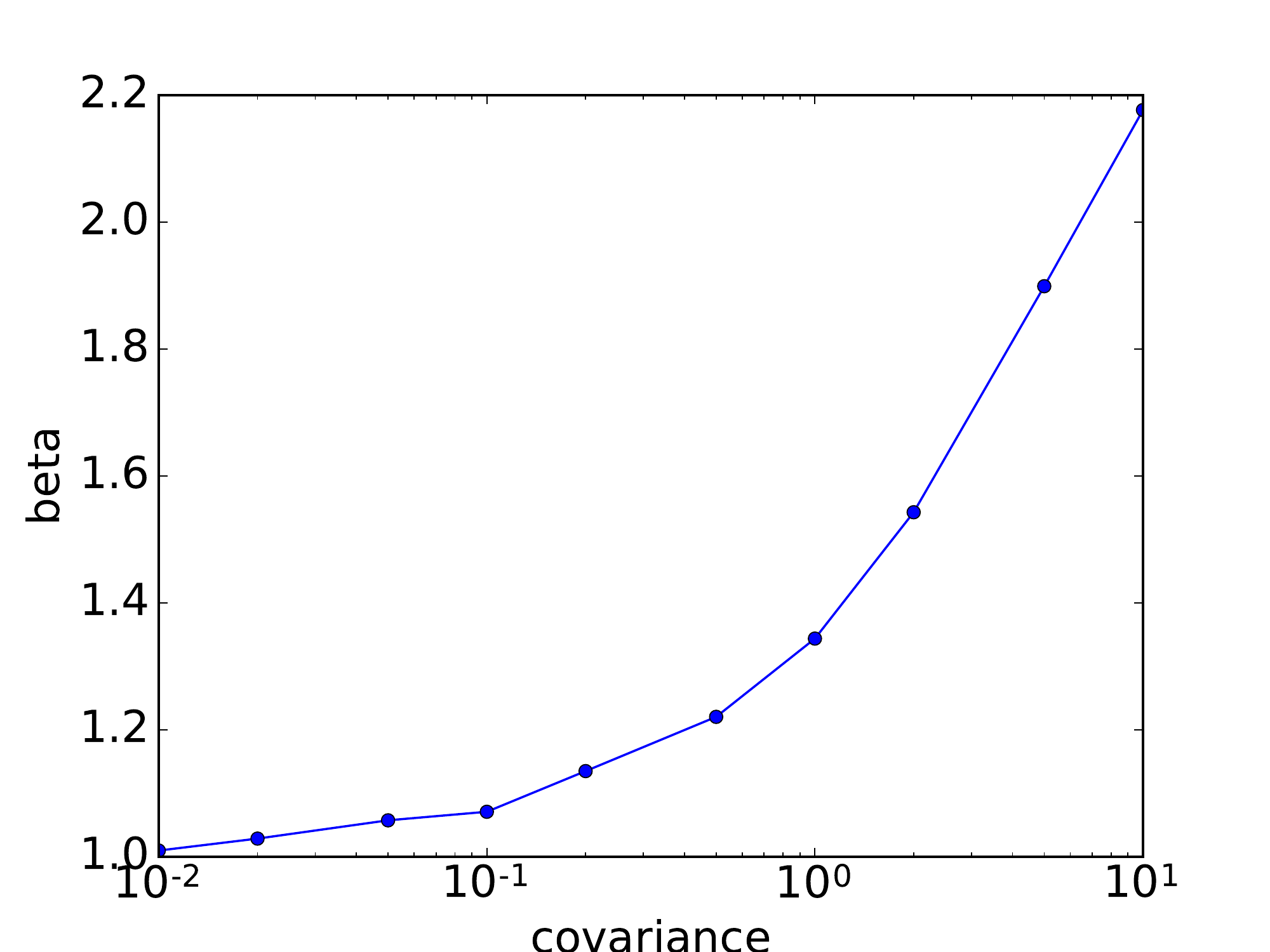}
    \caption{$k=5$, $d=2$.}
    \label{F:varvsratio_5}
  \end{subfigure}
  ~
  \begin{subfigure}[t]{0.5\textwidth}
  \includegraphics[scale=0.4]{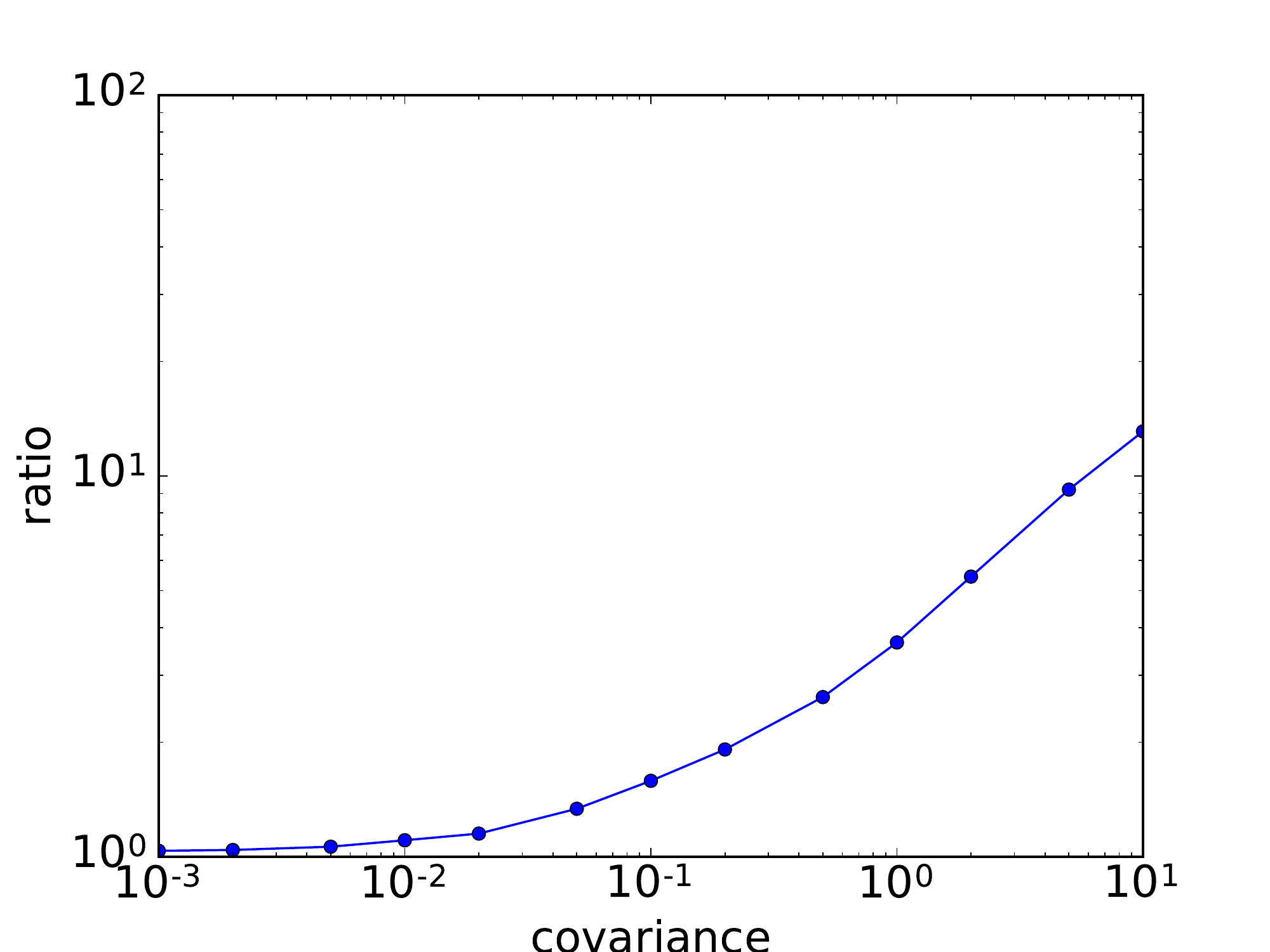}
    \caption{$k=50$, $d=2$.}
    \label{F:varvsratio_50}
  \end{subfigure}
  ~
    \begin{subfigure}[t]{0.5\textwidth}
      \includegraphics[scale=0.4]{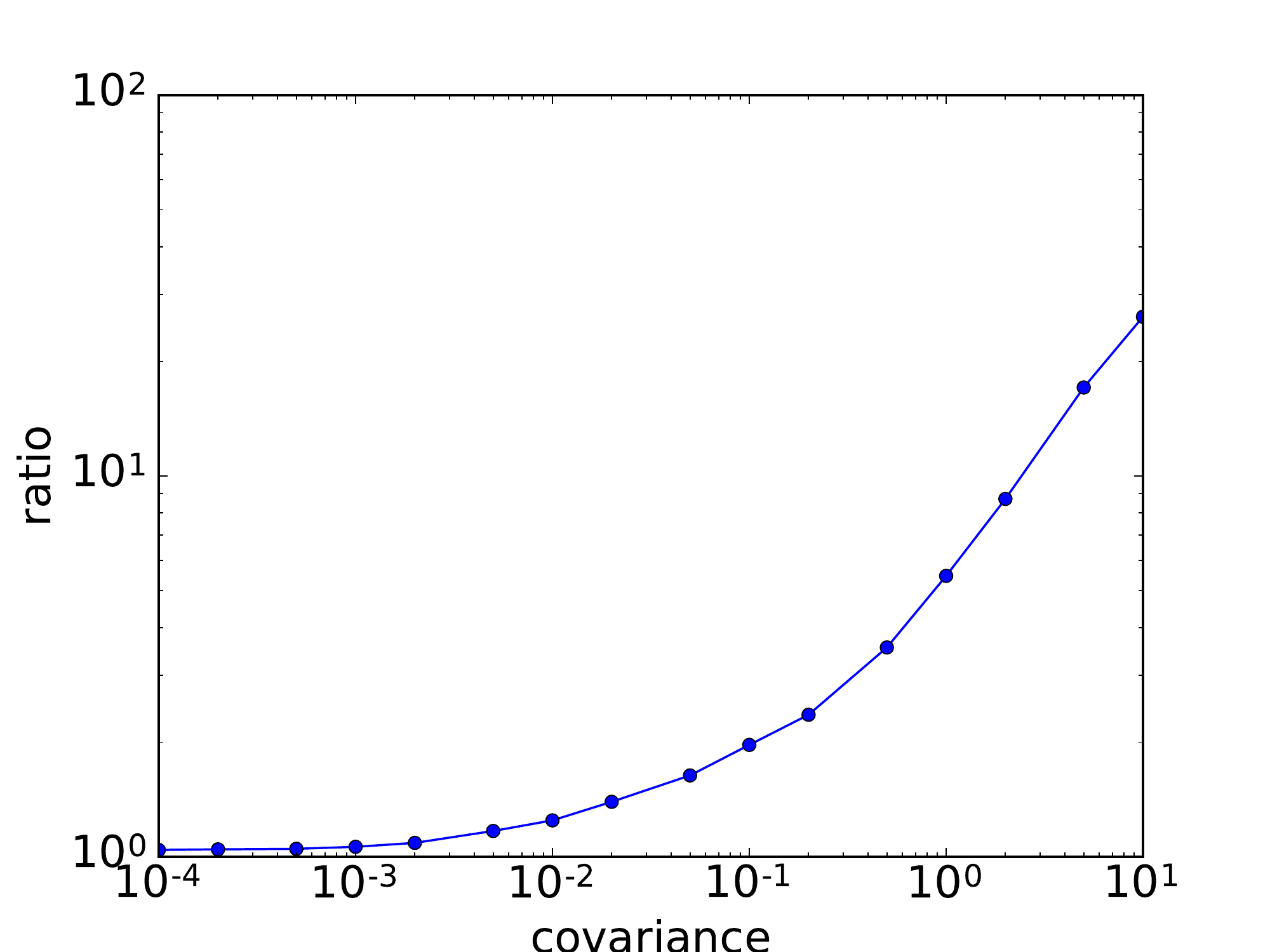}
      \caption{$k=100$, $d=2$.}
      \label{F:varvsratio_100}
    \end{subfigure}
  \caption{The ratio of the average $k$-means cost induced by the 
        means over the average optimal cost vs the variance for 
    2-dimensional instances generated from a mixture
    of $k$ Gaussians ($k \in \{5,50,100\}$).
    We observe that the $k$-means objective becomes ``relevant''
    (\ie is less than 1.05 times the optimal value)
    for finding the clustering induced by Gaussians when
    the variance is less than $0.1$ for $k=5$, less than $0.02$ 
    when $k=50$, and less than $0.0005$ when $k=100$.}
  \label{F:varvsratio_2d}
\end{figure}

The results are presented in Figures~\ref{F:varvsratio_2d} 
and~\ref{F:varvsratio_10d}.
We observe that when the variance $\sigma$ is large, the ratio
between the average value of the ground truth clustering and the 
average value of the optimal clustering becomes more important.
Indeed, the ground truth clusters start to overlap, allowing 
to improve the objective value by defining slightly 
different clusters. Therefore, the use of the $k$-means 
or $k$-median objectives
for modeling the recovery problem 
is not suitable anymore.
In these cases, since Local Search optimizes the solution with
respect to the current cost, the clustering output by local
search is very different from the ground truth clustering.
We thus identify instances for which the $k$-means objective 
is meaningful and so, Local Search is a relevant heuristic.
This motivates the following defintion.

\begin{defn}
  We say that a variance $\hat \sigma$ is \emph{relevant} if, for the
  $k$-means instances generated with variance $\hat \sigma$ the ratio
  between the average value of the ground truth clustering and the
  optimal clustering is less than 1.05.
\end{defn}

We summarize in Table~\ref{T:relev_var} the relevant variances
observed.
\begin{table}[h]
  \centering
  \begin{tabular}[h]{|l|l|c|r|}
    \hline
    \backslashbox{Number of dimensions}{Values of $k$}
                             &  $5$ & $50$ & $100$ \\
    \hline 
    \hline
    2  & $<0.05$ & $<0.002$ & $<0.0005$ \\
    \hline
    10 & $<15$ & $<1$ & $<0.5$\\
    \hline
    50 & $<1000000.0$ & $<100$ & $<7$ \\
    \hline
  \end{tabular}
  \caption{Relevant variances for $k \in \{5,50,100\}$ and 
    $d \in \{2,10,50\}$.}
  \label{T:relev_var}
\end{table}

\begin{figure}[h]
  \begin{subfigure}[t]{0.5\textwidth}
    \includegraphics[scale=0.4]{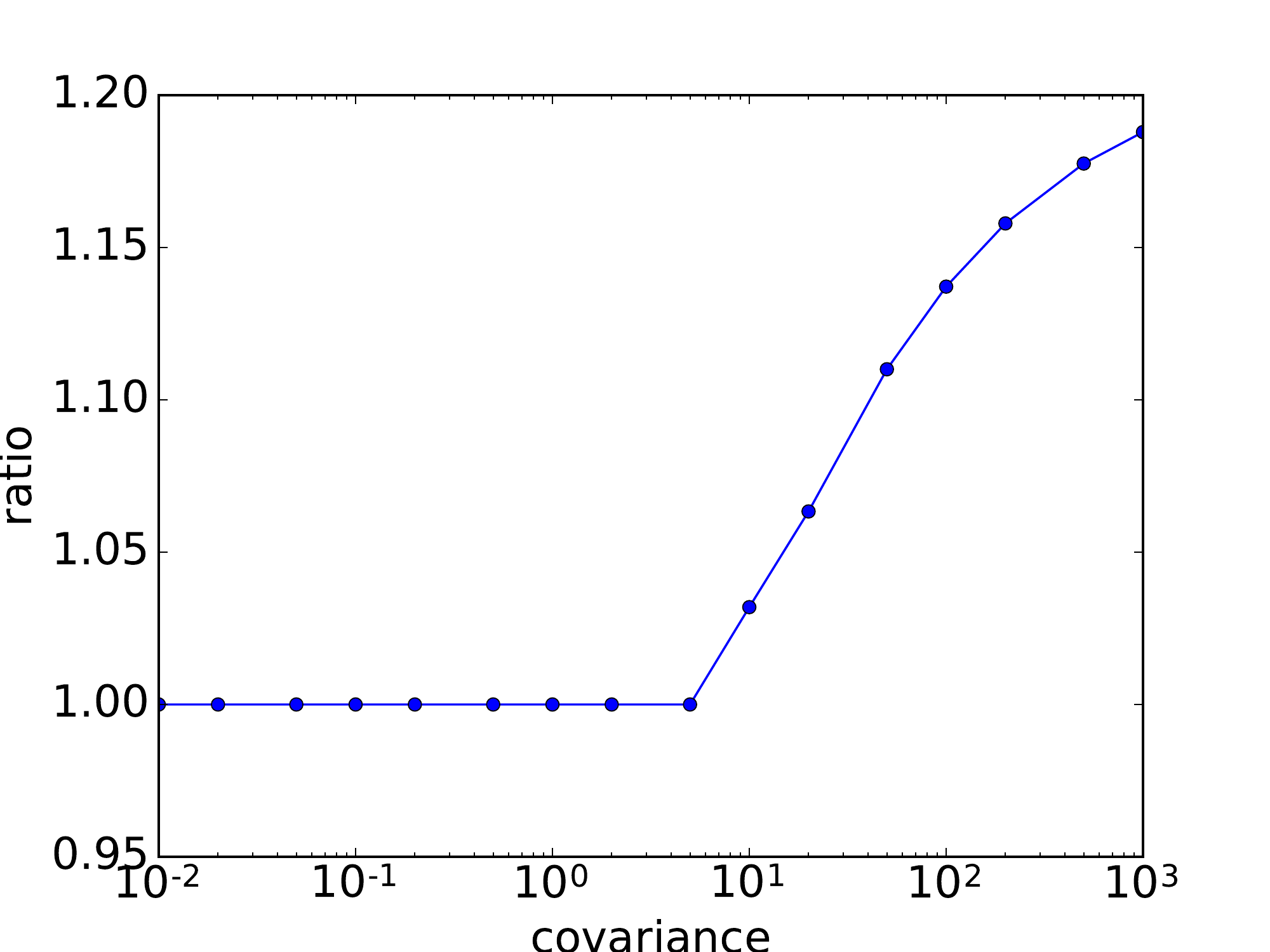}
    \caption{$k=5$, $d=10$.}
    \label{F:varvsratio_10d_5}
  \end{subfigure}
  ~
  \begin{subfigure}[t]{0.5\textwidth}
  \includegraphics[scale=0.4]{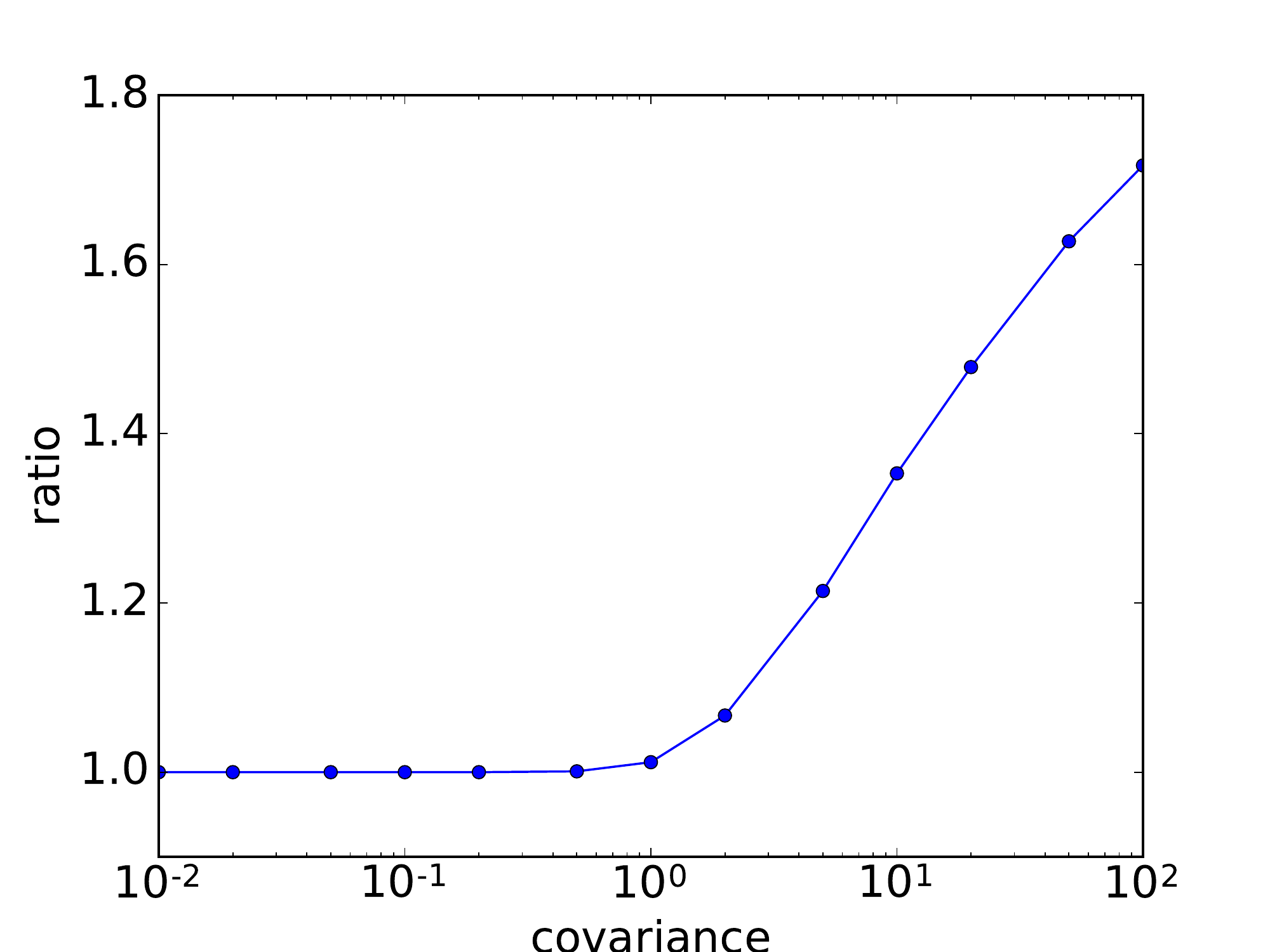}
    \caption{$k=50$, $d=10$.}
    \label{F:varvsratio_10d_50}
  \end{subfigure}
  ~
    \begin{subfigure}[t]{0.5\textwidth}
      \includegraphics[scale=0.4]{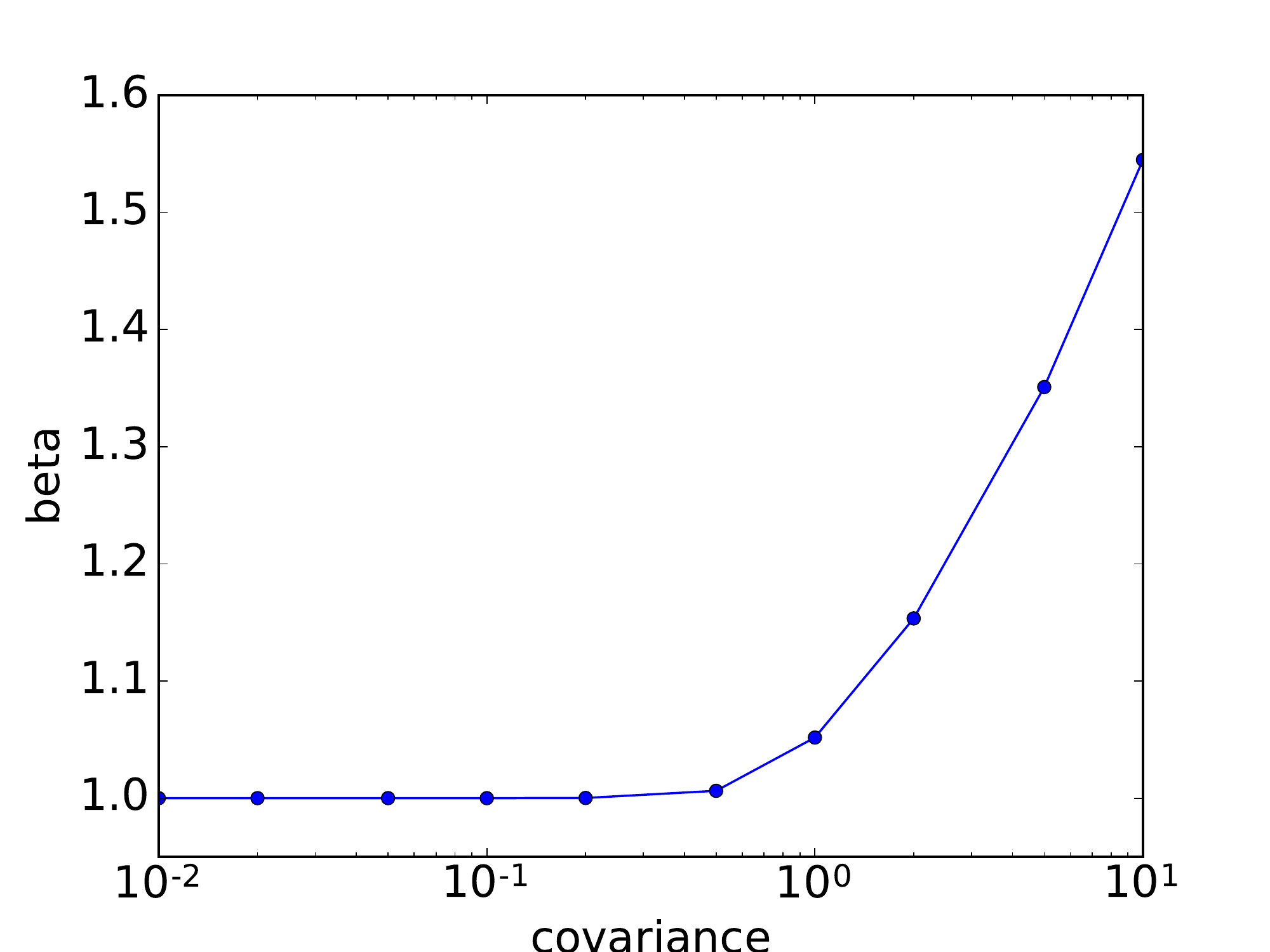}
      \caption{$k=100$, $d=10$.}
      \label{F:varvsratio_10d_100}
    \end{subfigure}
  \caption{The ratio of the average $k$-means cost induced by the 
    means over the average optimal cost vs the variance for 
    10-dimensional instances generated from a mixture
    of $k$ Gaussians ($k \in \{5,50,100\}$).
    We observe that the $k$-means objective becomes ``relevant''
    (\ie is less than 1.05 times the optimal value)
    for finding the clustering induced by Gaussians when
    the variance is less than $0.1$ for $k=5$, less than $0.02$ 
    when $k=50$, and less than $0.0005$ when $k=100$.}
  \label{F:varvsratio_10d}
\end{figure}

We consider the $\beta$-distribution-stability condition 
and ask whether the instances
generated from a relevant variance satisfy this condition for 
constant values of $\beta$.
We remark that $\beta$ can take arbitrarily small values.

We thus identify \emph{relevant} variances (see Table~\ref{T:relev_var})
for each pair $k,d$,
such that optimizing the
$k$-means objective in a $d$-dimensional instances generated 
from a relevant variance corresponds to finding 
the underlying clusters.







\paragraph{On stability conditions.} 
We now study the $\beta$-distribution-stability condition 
for random instances generated from a mixture of $k$ Gaussians.
The results are depicted in
Figures~\ref{F:varvsbeta_2d} and~\ref{F:varvsbeta_10d}.

\begin{figure}[h]
  \begin{subfigure}[t]{0.5\textwidth}
    \includegraphics[scale=0.4]{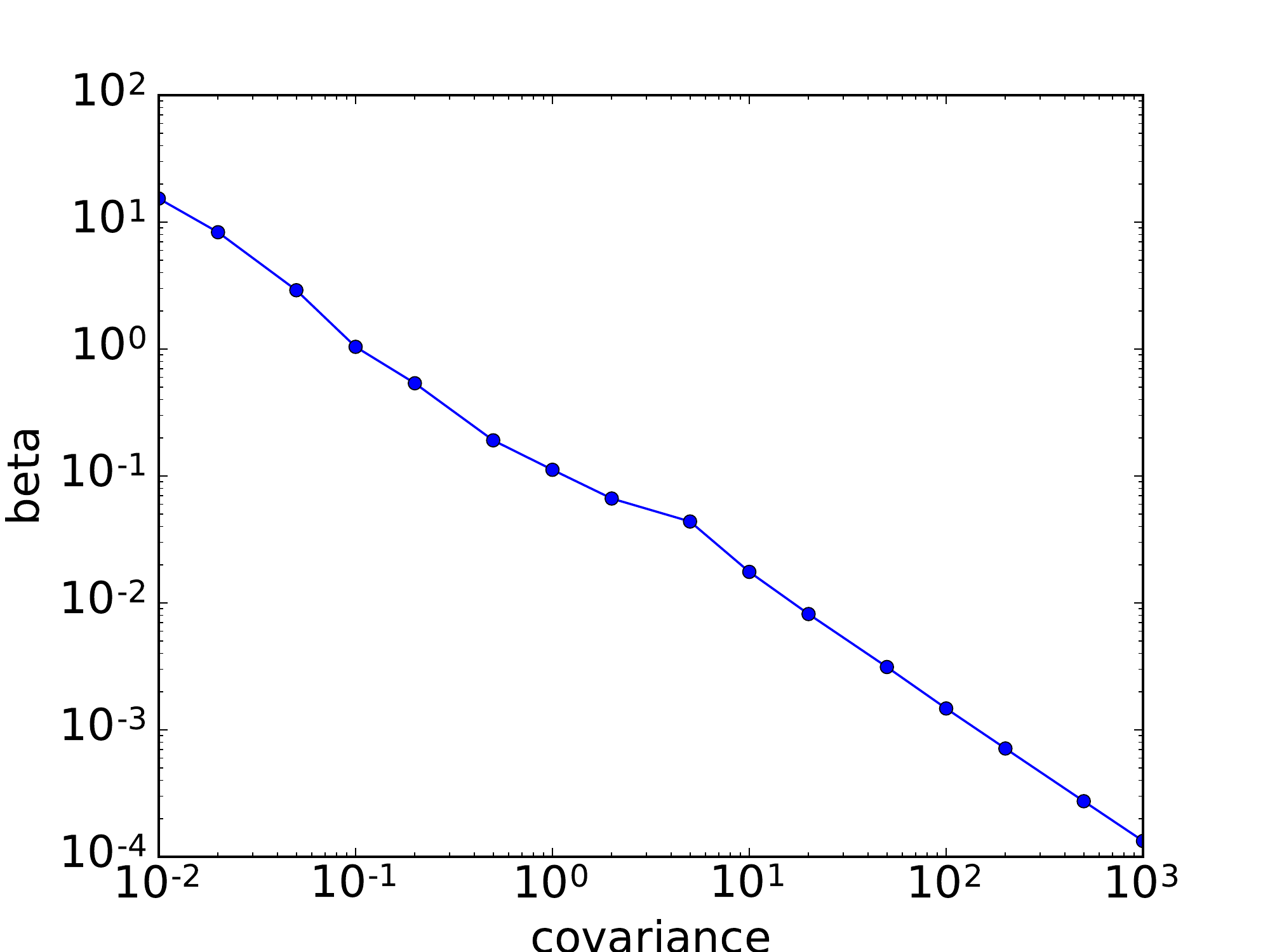}
    \caption{$k=5$, $d=10$.}
    \label{F:varvsbeta_10d_5}
  \end{subfigure}
  ~
  \begin{subfigure}[t]{0.5\textwidth}
  \includegraphics[scale=0.4]{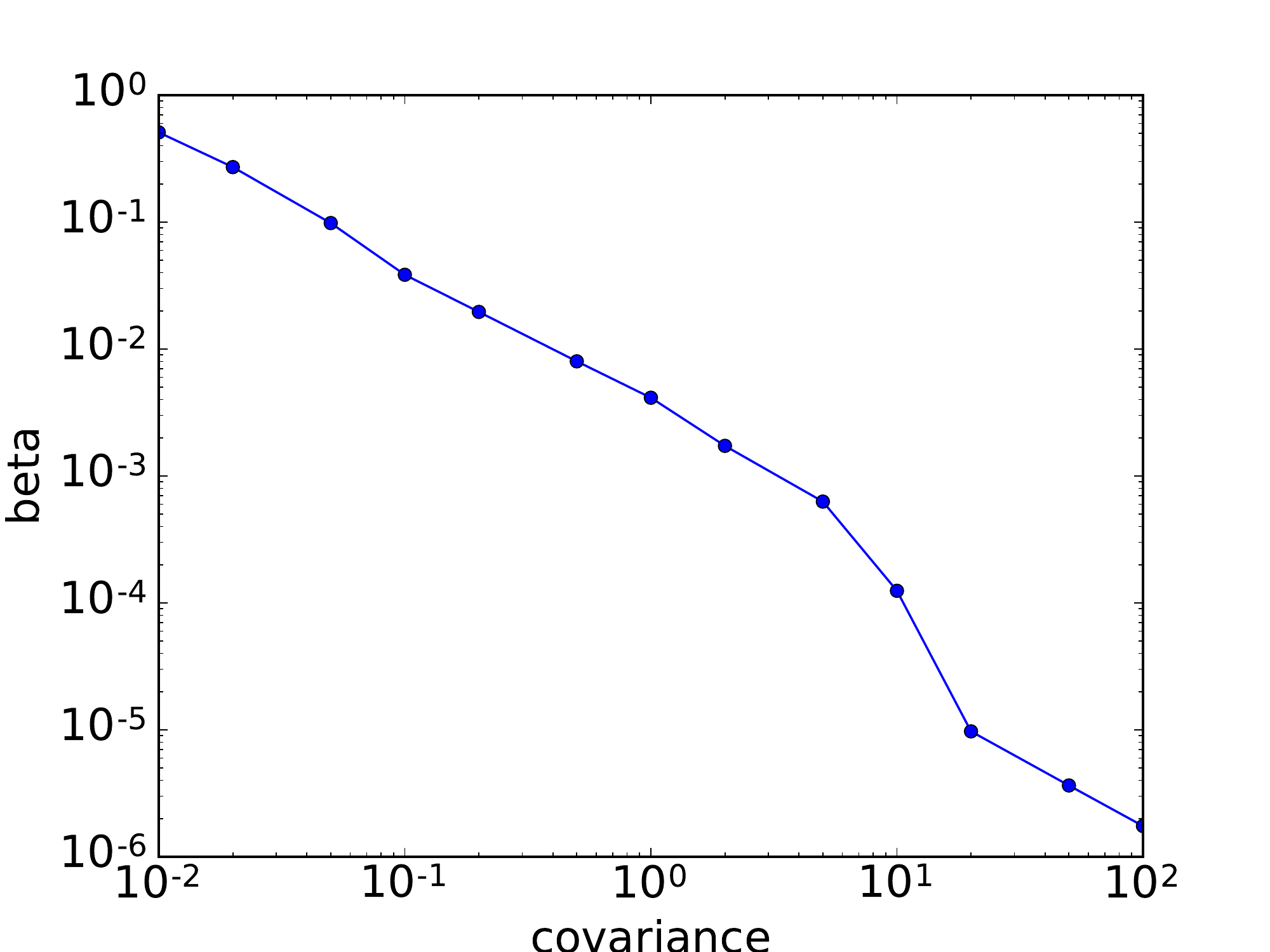}
    \caption{$k=50$, $d=10$.}
    \label{F:varvsbeta_10d_50}
  \end{subfigure}
  ~
  \begin{subfigure}[t]{0.5\textwidth}
    \includegraphics[scale=0.4]{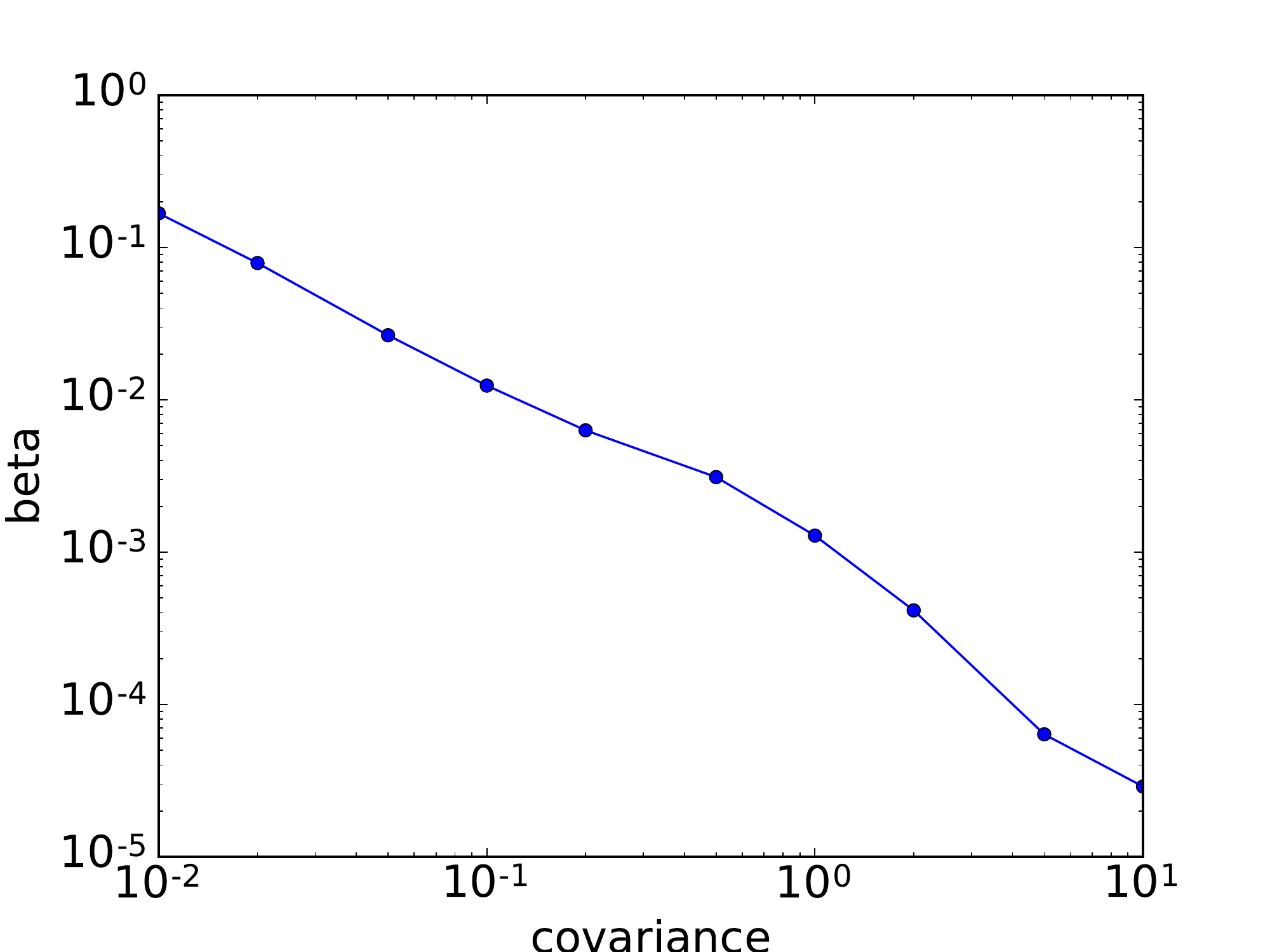}
    \caption{$k=100$, $d=10$.}
    \label{F:varvsbeta_10d_100}
  \end{subfigure}
  \caption{The average minimum value of $\beta$  for which
    the instance is $\beta$-distribution-stable vs the variance for 
    10-dimensional instances generated from a mixture
    of $k$ Gaussians ($k \in \{5,50,100\}$).
    We observe that for relevant variances, 
    the value of $\beta$ is greater than $0.001$.}
  \label{F:varvsbeta_10d}  
\end{figure}

\begin{figure}[h]
  \begin{subfigure}[t]{0.5\textwidth}
    \includegraphics[scale=0.4]{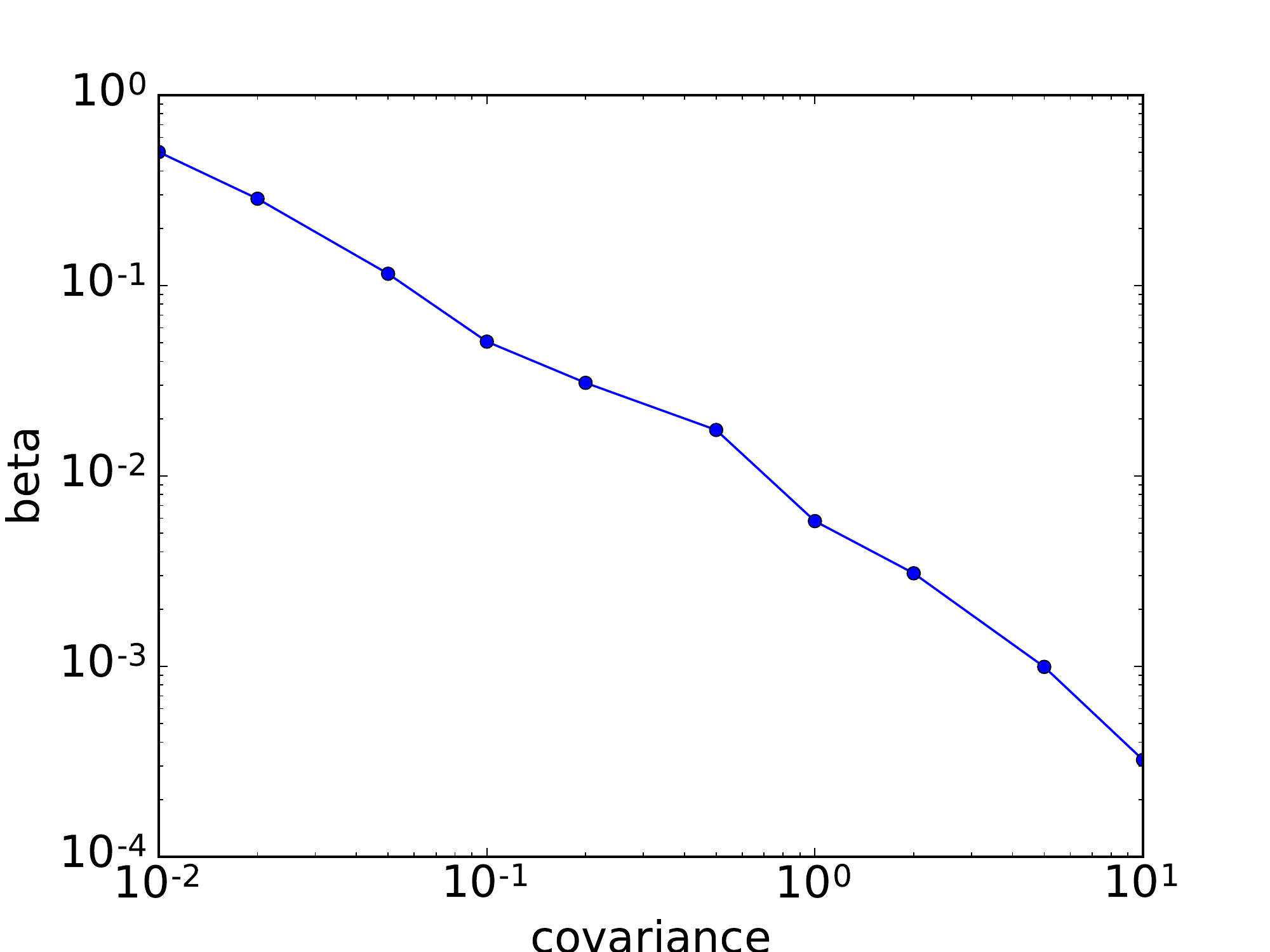}
    \caption{The average minimum value of $\beta$ for which
      the instances is $\beta$-distribution-stable
      vs the variance for 
      2-dimensional instances generated from a mixture
      of $5$ Gaussians.}
    \label{F:varvsbeta_5}
  \end{subfigure}
  ~
  \begin{subfigure}[t]{0.5\textwidth}
  \includegraphics[scale=0.4]{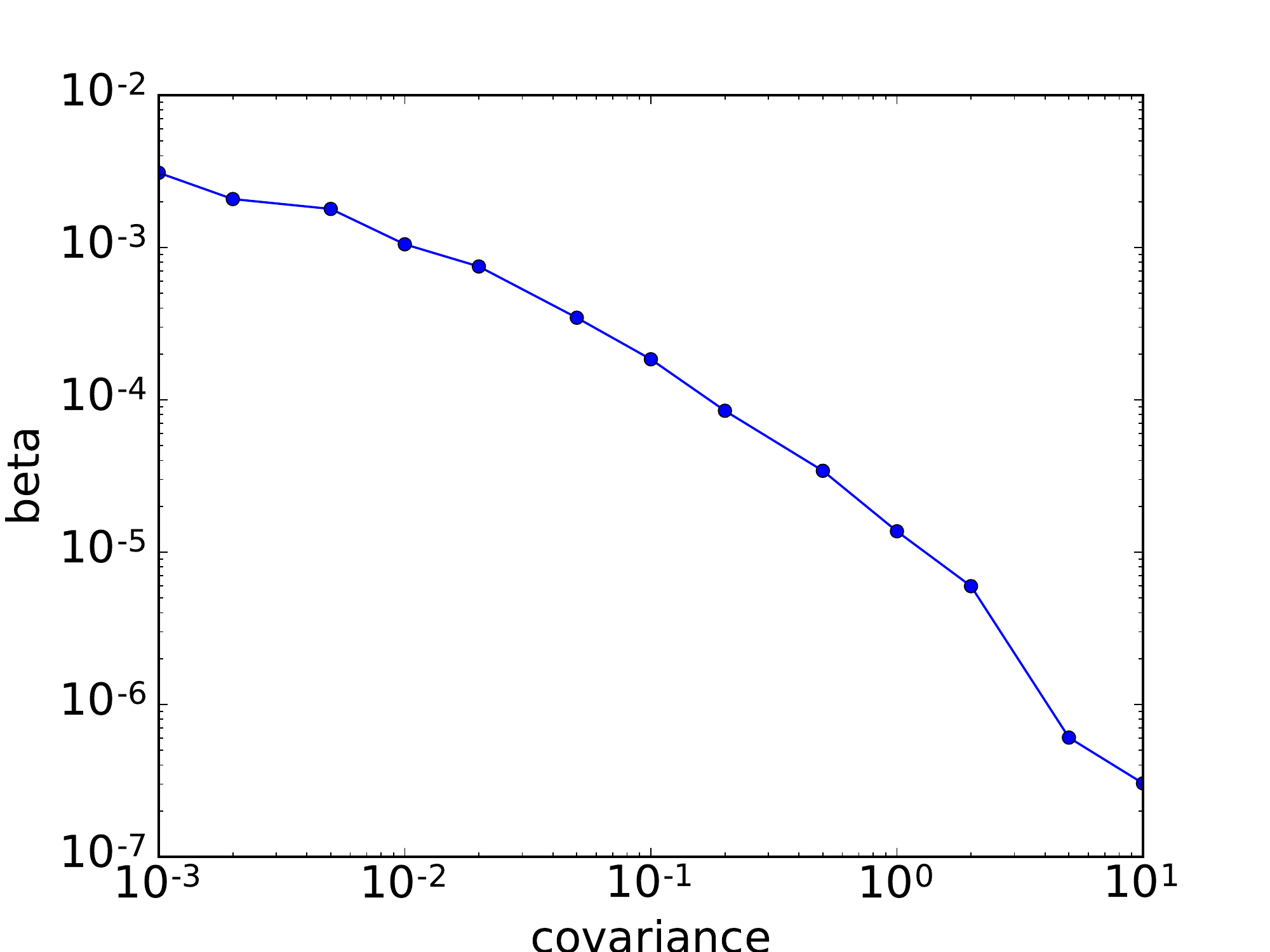}
    \caption{The average of the minimum value of $\beta$ for which
      the instances is $\beta$-distribution-stable
      vs the variance for 
      2-dimensional instances generated from a mixture
      of $50$ Gaussians.}
    \label{F:varvsbeta_50}
  \end{subfigure}
  ~
    \begin{subfigure}[t]{0.5\textwidth}
      \includegraphics[scale=0.4]{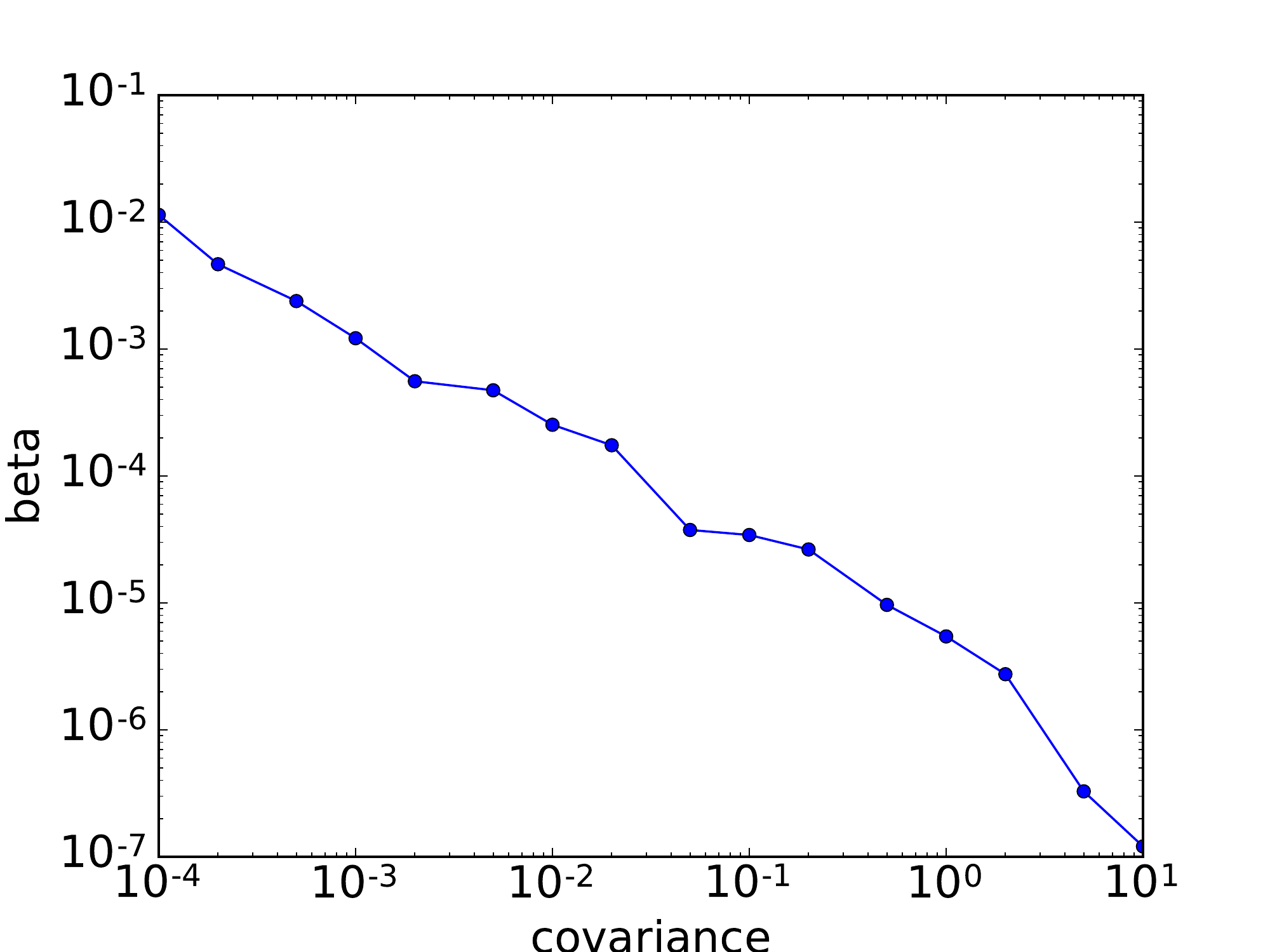}
      \caption{The average minimum value of $\beta$ for which
      the instance is $\beta$-distribution-stable vs the variance for
        2-dimensional instances generated from a mixture of $100$
        Gaussians.}
      \label{F:varvsbeta_100}
    \end{subfigure}
  \caption{The average minimum  value of $\beta$ for which
      the instance is $\beta$-distribution-stable vs the variance for 
    2-dimensional instances generated from a mixture
    of $k$ Gaussians ($k \in \{5,50,100\}$).
    We observe that for relevant variances, 
    the value of $\beta$ is greater than $0.001$.}
  \label{F:varvsbeta_2d}  
\end{figure} 


We observe that for random instances that are not generated from a
relevant variance, the instances are $\beta$-distribution-stable
for very small values of $\beta$ (\eg $\beta < 1e-07$).
We also make the following observation.
\begin{obs}\label{obs:rand:varbeta}
  Instances generated using
  relevant variances satisfy the $\beta$-distribution-stability
  condition for $\beta > 0.001$.  
\end{obs}

We remark that the number of dimensions is constant here and 
that having more dimensions might incur slightly different
values for $\beta$. It would be interesting to study this
dependency in a new study.

\bibliography{literature}{}
\bibliographystyle{plain}

\end{document}